\documentclass[sigconf]{acmart}

\AtBeginDocument{%
  \providecommand\BibTeX{{%
    \normalfont B\kern-0.5em{\scshape i\kern-0.25em b}\kern-0.8em\TeX}}}

\usepackage{subcaption}
\usepackage{amsmath,amsfonts, amsthm}
\usepackage{graphicx}
\usepackage{textcomp}
\usepackage{xcolor}
\usepackage{graphicx}
\usepackage{tabularx,ragged2e}
\usepackage[utf8]{inputenc}
\usepackage{kotex}
\usepackage{color}
\usepackage{url}
\usepackage{footnote}
\usepackage{float}
\usepackage{caption}
\usepackage{algpseudocode}
\usepackage{soul}
\usepackage{siunitx}
\usepackage[resetcount,linesnumbered,algoruled,boxed,lined]{algorithm2e}
\usepackage{multirow}
\usepackage{enumitem}
\usepackage[toc,page]{appendix}
\usepackage{hyperref}
\usepackage{comment}

\newcommand{\argmax}{\operatornamewithlimits{argmax}}

\newcommand{\rnbr}{\textsf{$r$-neighbors}}

\newcommand{\ernbr}{\textsf{$er$-neighbors}}
\newcommand{\BA}{\textsf{BA}}
\newcommand{\ERA}{\textsf{ERA}}
\newcommand{\FCA}{\textsf{FCA}}
\newcommand{\SEG}{\textsf{SEG}}
\newcommand{\AGA}{\textsf{AGA}}

\newcommand{\LUEM}{\textsf{LUEM}}

\SetCommentSty{mycommfont}
\SetKwInOut{Parameter}{parameter}
\SetKwFor{For}{for (}{) $\lbrace$}{$\rbrace$}

\newcolumntype{C}{>{\Centering\arraybackslash}X} 

\newtheorem{theorem}{Theorem}
\newtheorem{problemDefinition}{Problem definition}
\newtheorem{definition}{Definition}
\newtheorem{property}{Property}
\newtheorem{lemma}{Lemma}

\newtheorem{example}{Example}
\newtheorem{pDure}{Procedure}

\SetCommentSty{mycommfont}

\newcommand{\spara}[1]{\smallskip\noindent{\bf #1}}

\setcopyright{acmcopyright}
\copyrightyear{2018}
\acmYear{2018}
\acmDOI{XXXXXXX.XXXXXXX}

\acmConference[Conference acronym 'XX]{Make sure to enter the correct
  conference title from your rights confirmation emai}{June 03--05,
  2018}{Woodstock, NY}
\acmPrice{15.00}
\acmISBN{978-1-4503-XXXX-X/18/06}

\begin{document}

\title{LUEM : Local User Engagement Maximization in Networks}

\author{Junghoon Kim}
\affiliation{%
  \institution{UNIST}
  \country{South Korea}}
\email{junghoon.kim@unist.ac.kr}

\author{Jungeun Kim}
\affiliation{%
  \institution{Kongju National University}
  \country{South Korea}
}
\email{jekim@kongju.ac.kr}

\author{Hyun Ji Jeong}
\affiliation{%
  \institution{Korea Institute of Science and Technology Information}
  \country{South Korea}}
\email{hjjeong@kisti.re.kr}

\author{Sungsu Lim}
\affiliation{%
  \institution{Chungnam National University}
  \country{South Korea}}
\email{sungsu@cnu.ac.kr}

\renewcommand{\shortauthors}{Kim et al.}

\begin{abstract}
Understanding a social network is a fundamental  problem in social network analysis because of its numerous applications. Recently, user engagement in networks has received extensive attention from many research groups. However, most user engagement models focus on global user engagement to maximize (or minimize) the number of engaged users. In this study, we formulate the so-called  \underline{L}ocal \underline{U}ser \underline{E}ngagement \underline{M}aximization (\LUEM) problem. We prove that the {\LUEM} problem is NP-hard. To obtain high-quality results, we propose an approximation algorithm that  incorporates a traditional hill-climbing method. To improve efficiency, we propose an efficient pruning strategy while maintaining effectiveness. In addition, by observing the relationship between the degree and user engagement, we propose an efficient heuristic algorithm that preserves effectiveness. Finally, we conducted extensive experiments on ten real-world networks to demonstrate the superiority of the proposed algorithms. We observed that the proposed algorithm achieved up to 605\% more engaged users compared to the best baseline algorithms. 
\end{abstract}

\maketitle

\section{Introduction}\label{sec:introduction}
With the proliferation of mobile devices and the development of IT industry, many people use social networking services via mobile devices every day, anywhere, and anytime.
One of the biggest social networking services, Facebook, has reached over 2.8 billion monthly active users in 2021. Moreover,  approximately seven-in-ten U.S. adults use Facebook.  

With the increasing popularity of social networking services, understanding social networks is an important and fundamental problem~\cite{jin2013understanding,schneider2009understanding}. 
There have been several efforts to capture the characteristics of social networks, such as the importance of user nodes~\cite{bonacich1987power,borgatti2006graph,sabidussi1966centrality} or social relationships~\cite{brandes2001faster,de2012novel,mavroforakis2015spanning}, stability of a network~\cite{zhu2018k}, degree distribution~\cite{barabasi1999emergence}, small average distance~\cite{backstrom2012four}, and community size distribution~\cite{stephen2009explaining}. 

In addition, encouraging user engagement~\cite{malliaros2013stay,zhou2019k,bhawalkar2015preventing,ghafouri2020efficient,zhang2017finding,liu2021efficient,linghu2020global} in social networks has recently received extensive attention from many research groups. 
This study focused on engaged (activated) users in a social network. In practice, many services have many registered users but the amount of people who actively use the service is another matter. We can observe that many of the previously used services were shut down. Even if there are many registered users, if people do not actively use it, the service becomes meaningless.
The most widely accepted model~\cite{bhawalkar2012preventing,zhang2017engagement,linghu2020global,kim2022ocsm,cai2020anchored} for measuring the user engagement in a social network is based on the  \textit{minimum degree} ~\cite{seidman1983network}, which indicates that every engaged user has at least $k$ friends on a social network.
Formally, given a graph $G=(V,E)$ and positive integer $k$, if a user has at least $k$ friends, the user is considered to be \textit{engaged}. 
Otherwise, the user is disengaged from the network. 
If any user is disengaged, a set of users can iteratively be disengaged in a cascade manner because the number of friends of the remaining users changes~\cite{bhawalkar2015preventing, zhang2017finding, zhang2017olak}. Note that the number of engaged users in a social network can be identified by computing $k$-core~\cite{seidman1983network}. 
There are two main research directions regarding user engagement in social networks~\cite{bhawalkar2012preventing,chitnis2013preventing,zhang2017finding,luo2021parameterized,luo2021parameterized}. (1) \textit{The anchored $k$-core} problem was proposed by Bhawalkar et al.~\cite{bhawalkar2012preventing}. It aims to maximize the number of engaged users by anchoring $b$ disengaged users. In other words, the problem is to find $b$ important disengaged users. (2)  \textit{The collapsed $k$-core} problem was proposed by Zhang et al.~\cite{zhang2017finding}. It aims to minimize the number of engaged users by removing $b$ users, that is, the problem is to find $b$ important engaged users. 

\textcolor{black}{
\spara{Motivation.} The above user engagement problems focus on graph-level user engagement, that is, they aim to maximize (or minimize) the number of \textit{globally engaged users} in a network. This can be considered as a \textit{macro-level} user engagement.
To get all the engaged users, most approaches utilize existing cohesive subgraph models, such as $k$-core~\cite{seidman1983network} or $k$-truss~\cite{cohen2008trusses}. In the resultant cohesive subgraph, existing approaches try to find a set of anchored nodes~\cite{chitnis2013preventing} in the non-engaged users or collapsers~\cite{zhang2017finding} in the engaged users to maximize (or minimize) the number of engaged users globally, i.e., the total number of engaged or disengaged users is the major matter in the macro-level user engagement.   
Note that these approaches assume that a given network is a snapshot of a social network, i.e., each vertex is a real person and each edge indicates a friend relation in a social networking service. 
}

\textcolor{black}{
However, let assume that we want to perform targeted marketing for promotion, which is known as an effective marketing tool.
Targeted marketing involves breaking the target users into segments and then selecting influential users who widely spread our promotion.
In targeted marketing, selecting influential users for each segment is very important.
Since the target user is found by each segment, the user should have a local impact.
In these days, according to the proliferation of social networks, many companies use social networks to find influential users.
Hence, we propose the problem that finds a kind of micro-level (or called local) user engagement from social networks.
}

\textcolor{black}{
This problem is distinguished from the anchored $k$-core problem~\cite{chitnis2013preventing} and collapsed $k$-core problem~\cite{zhang2017finding}. The following example describes that our problem can be utilized to promote a social networking service.   
}
\textcolor{black}{
\begin{example}
By utilizing the {\LUEM} problem, we can find a set of people who effectively promote the products to their friends. The selected people may have many friends and there are sufficient social relationships among friends, i.e., they are important from the perspective of user engagement. This approach is distinguished from the degree-based approach. If we consider only the degree of the nodes, a set of fake accounts which has many fake friends will be selected. These fake accounts have no effect on increasing the number of engaged users since they may have only few relationships to each other. 
\end{example}
}

As an alternative to the difficult-to-expect effects in global user engagement, we focused on local user engagement, which unifies three key concepts: (1) The existence of seed nodes, (2) minimum-degree-based user engagement;  and (3) seed-based distance cohesiveness. These key concepts are described as follows:


\begin{figure}[t]
\centering
\includegraphics[width=0.95\linewidth]{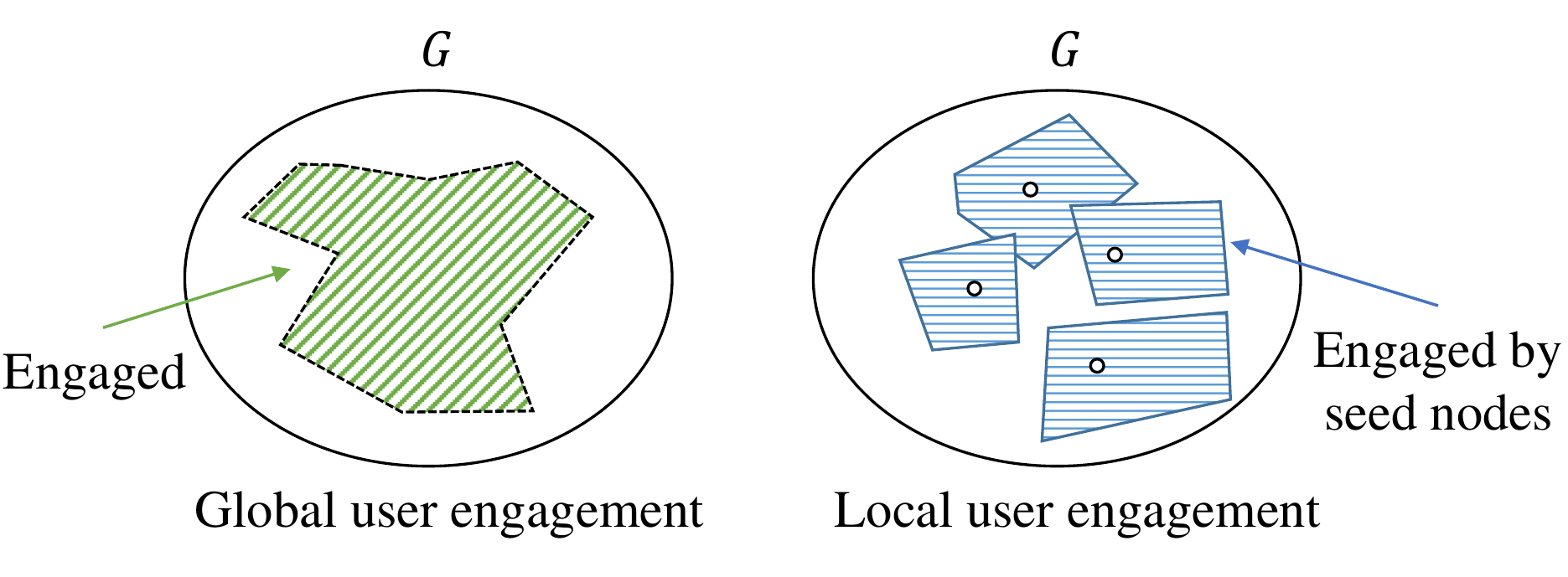}
\vspace{-0.3cm}
\caption{Global user engagement and local user engagement}
\label{fig:twoEngagement}
\end{figure}

\begin{itemize}[leftmargin=*]
    \item \textit{Existence of seed nodes} : We focused on the local structure of seed nodes to measure user engagement. This changes the perspective from a global network structure to a local network structure. Intuitively, when a node is selected as the seed node, a set of nodes closely related to the seed node can be engaged. 
    This can be considered as a micro-level of user engagement. This has led to new applications. Figure~\ref{fig:twoEngagement} depicts the difference between the global user engagement problem (macro-level, green-coloured) and proposed local user engagement problem (micro-level, blue-coloured). Notice that local user engagement has seed nodes, and the engaged users are affected by the seed nodes and can overlap. 
    There have been several similar attempts at social network analysis. The influence maximization (IM) problem~\cite{kempe2003maximizing} aims to identify a set of seed nodes to maximize the spread of influence in a social network under information diffusion models.
    \item \textit{Minimum-degree-based structural cohesiveness} : We incorporated the widely used minimum-degree user engagement model~\cite{bhawalkar2012preventing,zhang2017olak,zhang2017finding,zhang2018finding} to measure  user engagement. Formally, a user is engaged if he/she has at least $k$ friends in the local structure. Otherwise, the user is disengaged. 
    \item \textit{Distance-generalized cohesiveness} : In the social sciences field, studying the 1-hop neighbor structure to capture  structural characteristics is well established~\cite{luce1950connectivity}. For the user engagement-related problem, the 1-hop neighbor structure can be considered because many social networks are scale-free~\cite{barabasi1999emergence}; 
    thus, the degree of the nodes in social networks follows a power-law distribution. Therefore, we generalize the relationship between two nodes by incorporating the graph distance. i.e., if a node is reachable from the seed node within a specific user-defined distance threshold, the node can be engaged. 
\end{itemize}

In this paper, we aim to find an answer to the question \textit{How many close neighbor nodes will be engaged when a user is activated (becomes a seed node)?} This question is answered by proposing a new information diffusion model based on user engagement in the IM problem. Compared to the traditional information diffusion models, our model is characterized by a spreading limitation (distance) and influenced limitation (minimum degree) without using any simulation models. Hence, the results are deterministic. Figure~\ref{fig:intro_sum} illustrates a summary of our problem.


\begin{figure}[t]
\centering
\includegraphics[width=0.9\linewidth]{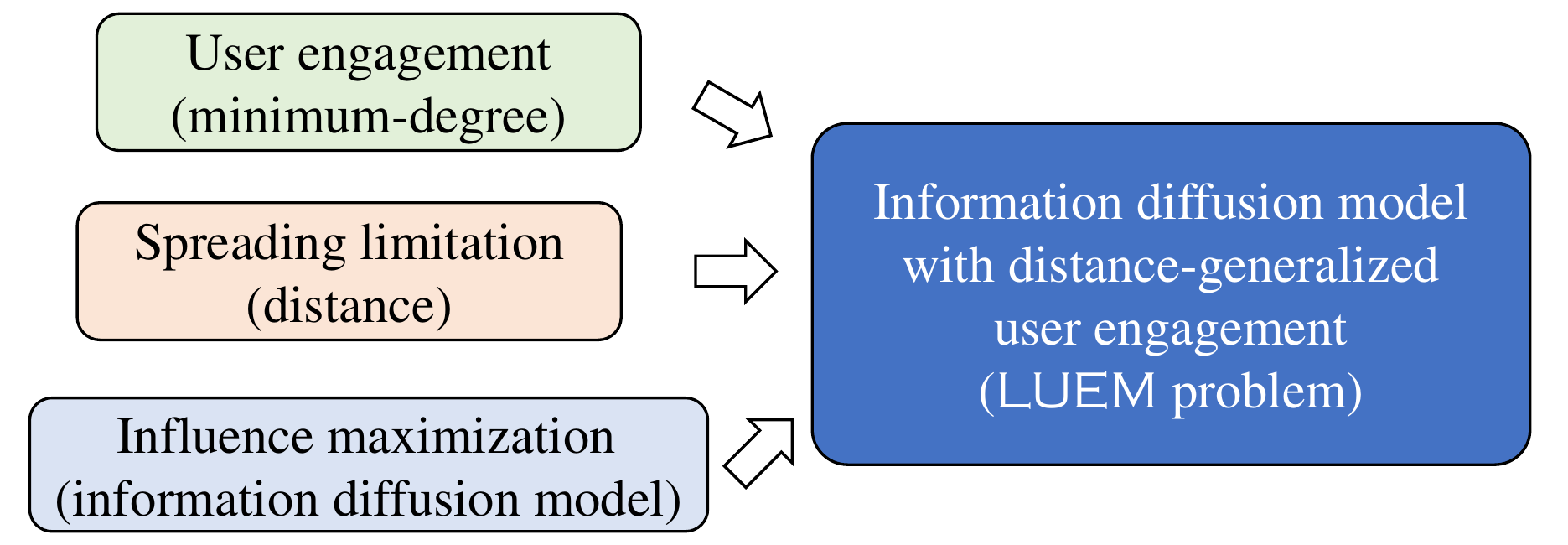}
\caption{Illustration of our problem}
\label{fig:intro_sum}
\end{figure}
We formulate a local user engagement maximization ({\LUEM}) problem. First, we propose a seed-based engaged group({\SEG}) for a specific node $v$. 
A set of nodes $H$ is called the {\SEG} of the node $v$ if the nodes in $H$ are engaged by the seed node $v$. That is, the minimum degree of the induced subgraph of nodes $H$ is larger than or equal to threshold $k$, and all nodes in $H$ are reachable from $v$ at distance $r$. 
The {\LUEM} problem is defined as follows. Given a graph $G$, budget $b$, engagement threshold $k$, and distance threshold $r$, {\LUEM} aims to find $b$ seed nodes that can maximize the number of engaged users who have at least $k$ friends in each {\SEG} and are reachable within a distance $r$ from the seed nodes.
We show that this problem is NP-hard, and the objective function is monotone submodular. Therefore, we present an approximation algorithm that holds $1-1/e$-approximation ratio. We observe that the main bottleneck of the proposed algorithm is to compute the {\SEG}s of the nodes because for each {\SEG}, it is necessary to compute the $k$-core and connected component containing a seed node ($O(|V|+|E|)$). Thus, to improve efficiency, we propose a new heuristic algorithm that incorporates the approximated neighborhood functions~\cite{boldi2011hyperanf} without computing all {\SEG}s at the initial stage.

\spara{Applications.}
The applications of our proposed problem are listed as follows:  

\begin{itemize}[leftmargin=*]
    \item \textit{Organizing a party.} Let us assume that Joy wants to host a farewell party for three consecutive holiday days by serving dinner. She would like to invite three of her friend groups. Because her friends are from different groups, such as a table tennis club, graduate association club, lab colleagues, etc., she would like to invite as many non-engaged people in the party as possible. Here, finding a few friends to help her organize party.
    \item \textit{Finding top-$b$ influential communities.} Finding the engaged users of a specific node can be considered as  finding an influential community of the specific node. 
    \item \textit{Finding key users in a social network.} Important users are nodes that can affect the network structure to a greater extent compared to other nodes~\cite{xiaolong2014review}. Our {\LUEM} can be used to identify important users in a social network using the user engagement model.
    \item \textit{Virus propagation prediction/estimating influence spreading.} \textcolor{black}{ {\LUEM} can identify a set of engaged users of a specific node, and thus, it can be utilized for predicting the virus propagation and estimating the influence spreading. If a user is infected, a set of close neighbor nodes has \textcolor{black}{a} high probability of getting infected.  }
\end{itemize}

\spara{Challenges and Contributions}
Because the {\LUEM} problem is NP-hard, computing an exact solution within polynomial time is prohibited. Thus, the first challenge is to compute an effective solution. The second challenge is the efficient computation of the solution. To address these challenges, we propose two algorithms: (1) An effective $r$-neighbor-based approximation algorithm (\ERA)  and (2) Fast circle algorithm (\FCA). First, the {\ERA} algorithm incorporates a traditional greedy algorithm to maximize the number of engaged users. Owing to the submodularity of our objective function, we can design an efficient strategy by not computing all engaged users for every iteration. {\ERA} holds a $1-1/e$ approximation ratio.  In section~\ref{sec:approx}, we present our main idea and rigorous proof to demonstrate  why the pruning strategy preserves the approximation ratio.
Next, we propose a heuristic algorithm called {\FCA} to improve the efficiency of the proposed algorithms. Even if our proposed {\ERA} significantly improves efficiency, in the worst case, it has the same time complexity as a traditional greedy algorithm. Thus, in {\FCA}, we propose a very fast heuristic algorithm that incorporates the approximated neighborhood function based on our observations.

The contributions of this research are summarized as follows: 

\begin{itemize}[leftmargin=*]
    \item  \textit{Problem definition : } To the best of our knowledge, this is the first study to identify a set of seed nodes to maximize the number of engaged users in a social network. 
    \item \textit{Theoretical analysis : } We prove that the objective function of {\LUEM}  is monotone submodular, and the {\LUEM} problem is NP-hard.  
    \item \textit{Designing new algorithms : } Because our problem is NP-hard, we propose a $1-1/e$ approximation algorithm, as well as a heuristic algorithm to improve efficiency. 
    \item \textit{Extensive experimental study : } Using real-world datasets, we conduct extensive experiments to demonstrate the superiority of the proposed algorithms. 
\end{itemize}

\section{Problem Statements}

We present a \underline{L}ocal \underline{U}ser \underline{E}ngagement \underline{M}aximization ({\LUEM}) problem. In this study, we consider an unweighted and undirected graph. Given a graph $G=(V,E)$ and set of nodes $H\subseteq V$, we denote $G[H]=(H, E[H])$ as a subgraph of $G$ induced by nodes $H$. 
Table~\ref{tab:notation} lists the basic notations used in this study. First, We introduce basic definitions of some terminology for presenting our problem. 

\begin{table}[t]
\vspace{-0.2cm}
\caption{Notation}
\centering
\label{tab:notation}
\begin{tabular}{c|c}
\hline
Description                     & Notation                  \\ \hline \hline
minimum degree threshold       & $k$                       \\ \hline
distance threshold              & $r$                       \\ \hline
user engagement constraint      & $\delta(.)$                 \\ \hline
distance constraint             & $\tau(.)$                 \\ \hline
approximated neighbor value    & $ANV$                     \\ \hline
{\rnbr} of node $v$ in $G$      & $N_r(v,G)$                \\ \hline
effective $r$-neighbors size of node $v$       & $\mathcal{E}_{r}(v)$   \\ \hline
{\SEG} of node $v$      & $\mathcal{K}_{k,r}(v)$                \\ \hline
distinct engaged users in $\mathcal{K}$ & {$\rho$}($S$)   \\ \hline
engagement gain & {$\rho$}($v$, $S$)   \\ \hline
\hline
\end{tabular}  
\end{table}

\begin{definition}
(\underline{User engagement constraint $\delta(.)$})\\
Given a graph $G=(V,E)$ and positive integer $k$ called the minimum degree threshold, a subgraph $H\subseteq V$ satisfies the user engagement constraint if it is connected, and the minimum degree of induced subgraph $G[H]$ is larger than or equal to $k$, i.e., $\delta(H)\geq k$.
\end{definition}

Given a graph $G=(V,E)$, finding a maximal subgraph satisfying the minimum degree constraint is the same as that in the classic   $k$-core~\cite{seidman1983network} problem. $k$-core can be computed in polynomial time. 


\begin{definition}
(\underline{Distance constraint $\tau(.)$})\\
Given a graph $G=(V,E)$, seed node $s\in V$, and positive integer $r$ called a distance threshold, a subgraph $H\subseteq G$ satisfies the distance constraint if it contains $s$, and the distance from $s$ to any node in the induced subgraph $G[H]$ is less than or equal to $r$, that is, $\tau(s, H)\leq r$. 
\end{definition}

We next define {\rnbr} based on the distance threshold $r$. 

\begin{definition}
(\underline{{\rnbr} $N_r$})\\
Given a graph $G=(V,E)$, node $v\in V$, and positive integer $r$, {\rnbr} of $v$, denoted as $N_r(v, G)$, is a set of nodes that are reachable from $v$ within a distance $r$ in graph $G$. We use $N(v,G) = N_1(v,G)$, and $N(v)$ if it is obvious. 
\end{definition}

We are now ready to discuss user-level engagement by defining a set of users engaged by the seed node \textcolor{black}{as follows}.

\begin{definition}
(\underline{Seed-based engaged group(\SEG)}) \\
Given a graph $G=(V,E)$, seed node $s$, and positive integers $k$ and $r$, seed-based engaged group of the node $s$, denoted as $\mathcal{K}_{k,r}(s)$, is a maximal set of nodes satisfying $\delta(\mathcal{K}_{k,r}(s))\geq k$ and $\tau(s, \mathcal{K}_{k,r}(s))\leq r$. If it is obvious, we use $\mathcal{K}(s)$ instead of $\mathcal{K}_{k,r}(s)$. 
\end{definition}

Note that a set of users in $\mathcal{K}(s)$ is considered as \textit{engaged} by the seed node $s$.
As we have discussed in section~\ref{sec:introduction}, we aim at maximizing the number of engaged users by selecting seed nodes. 
We next define the engagement gain of {\SEG}.
\begin{definition}
(\underline{Engagement gain}) \\
Given a graph $G=(V,E)$, and a set of seed nodes $\mathcal{S}$, and a new seed node $v\not \in \mathcal{S}$, the engagement gain $\rho(v, \mathcal{S})$ is the number of newly engaged users in $\mathcal{K}(v)$, i.e.,  $|\mathcal{K}(v) \setminus \bigcup_{s\in \mathcal{S}} \mathcal{K}(s)|$. 
\end{definition}
We formally define our objective function named local user engagement function. 

\begin{definition}
(\underline{Local user engagement function  $\rho(.)$}). \\
Given a set of seed nodes $\mathcal{S}$, local user engagement function $\rho (\mathcal{S})$ returns the number of distinct engaged users, namely, $|\bigcup_{s \in \mathcal{S}} \mathcal{K}_{k,r}(s)|$ . 
\end{definition}

\begin{example}
In Figure~\ref{fig:twoKrcc}, suppose that $k=2$ and $r=2$. We can check two {\SEG}s : $\mathcal{K}_{2,2}(b)$ and $\mathcal{K}_{2,2}(f)$. Suppose that a set of seed nodes $\mathcal{S}=\varnothing$. Notice that  $\rho(b,\mathcal{S})=3$ and $\rho(f,\mathcal{S})=5$. However, when  $\mathcal{S}=\{f\}$, $\rho(b,\mathcal{S})$ is changed from $3$ to $2$ because a node $c$ is already engaged owing to the seed node $d$. 
\end{example}

\begin{figure}[t]
\centering
\includegraphics[width=0.9\linewidth]{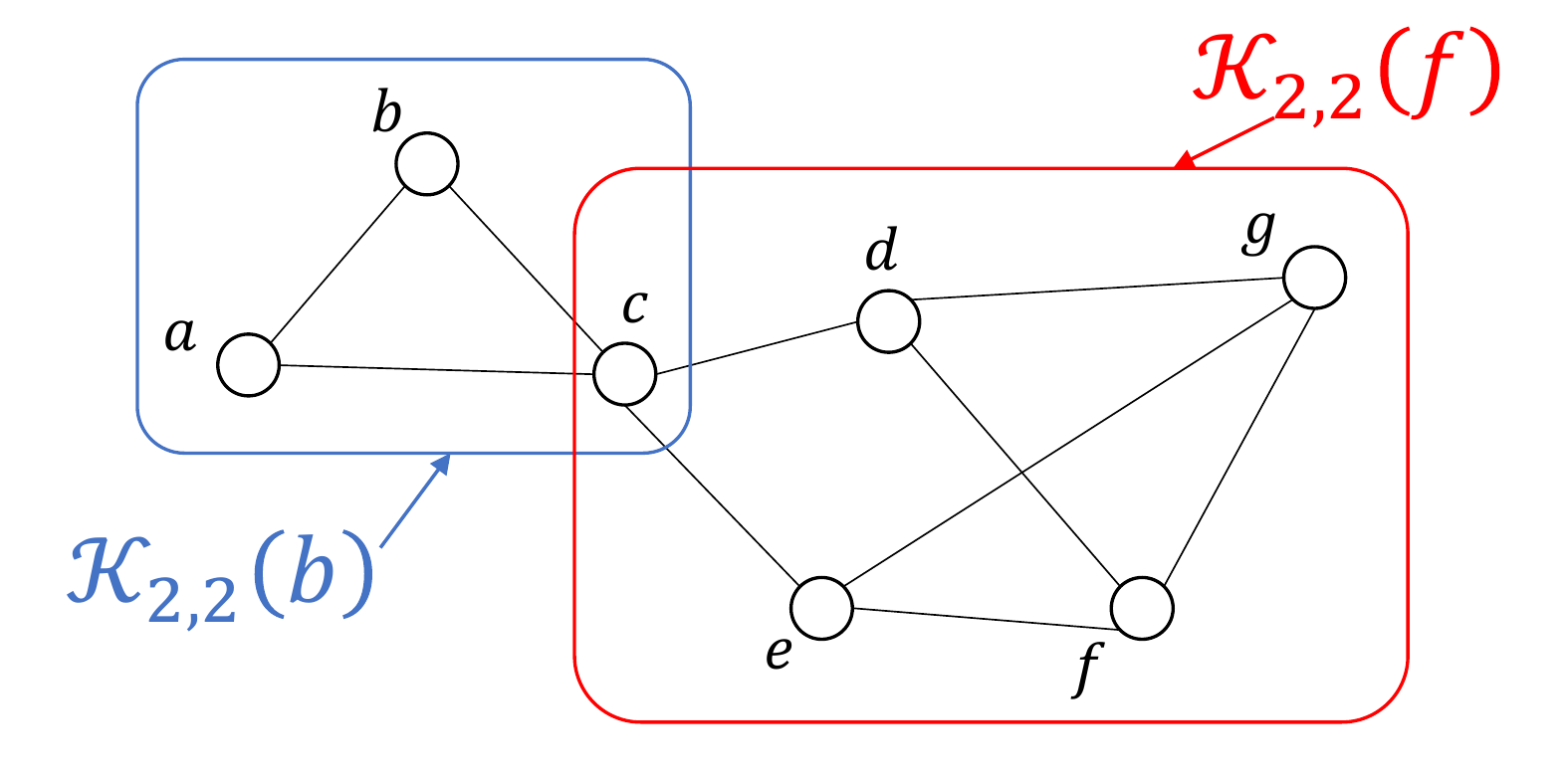}
\vspace{-0.3cm}
\caption{Two {\SEG}s when $k=2$ and $r=2$}
\label{fig:twoKrcc}
\end{figure}

\begin{property}\label{property:kcorecontain}
Given a graph $G=(V,E)$, $k$, and $r$, the {\SEG} of any node $v\in V$ always belongs to $k$-core and is unique. 
\end{property}

\begin{proof}
Because {\SEG} is a maximal cohesive subgraph within a distance $r$ from the seed node $s$, it is unique. When $r=\infty$, the result of any {\SEG} belongs to the $k$-core because the $k$-core does not require the connectivity constraint of the resultant subgraphs. 
\end{proof}

Note that $\mathcal{K}_{k,r}(v) \subseteq  N_r(v)$. Observe that the union of all the {\SEG}s in a network is the same as $k$-core. This implies that the maximum number of engaged users of given networks can be computed in an efficient way. 
As discussed in the applications, we are interested in finding $b$ seed nodes to maximize the number of engaged users. Now we can formulate our problem. The definition of our {\LUEM} problem is as follows. 

\begin{problemDefinition}
(Local User Engagement Maximization (\LUEM)). 
Given a graph $G=(V,E)$, positive integers $b$, $k$, and $r$, the \LUEM\ problem aims to identify  $b$ seed nodes, denoted as $\mathcal{S}$, such that the number of distinct engaged users is maximized; in other words, maximizing $\rho(\mathcal{S})$ such that $|\mathcal{S}|\leq b$. 
\end{problemDefinition}

Note that any pair of {\SEG}s can overlap. Next, we present some of the important properties of the {\LUEM}  problem. 

\begin{property}
The local user engagement function $\rho(.)$ is submodular.
\end{property}

\begin{proof}

It is known that a function $g$ is submodular if for all $S\subseteq T\subseteq V$, all $v\in V\setminus T$,
\textcolor{black}{$g(S\cup \{v\})-g(S)\geq g(T\cup \{v\})-g(T)$~\cite{schrijver2003combinatorial}}
. 
Assume that $S\subseteq T\subseteq V$ and  there is a node $u\in V\setminus T$ which makes
\textcolor{black}{$\rho(S\cup \{u\})-\rho(S)< \rho(T\cup \{u\})-\rho(T)$}
. It implies that $\rho(u, T)$  is larger than  $\rho(u, S)$.
However, we know that (1) $\rho(T)\geq \rho(S)$; and (2) a user engaged by the set of seed nodes $S$ is always the engaged user by the set of seed nodes $T$. 
Hence, the number of engageable users by the node $u$ in $T$ is smaller than or equal to the number of engageable users by the node $u$ in $S$. 
Hence, due to $S \subseteq T$, 
\textcolor{black}{
$\rho(S\cup \{u\})-\rho(S)< \rho(T\cup \{u\})-\rho(T)$
}
does not hold. It implies that our assumption is not true, and our function $\rho(.)$ is submodular. 
\end{proof}

\begin{property}
$\rho(.)$ function is monotone.
\end{property} 

\begin{proof}
The proof is trivial. For any $S \subseteq T \subseteq V$, $\rho(S) \leq \rho(T)$ always holds. Therefore, {$\rho(.)$} is monotone. 
\end{proof}

\begin{theorem}\label{theorem:NP}
{\LUEM} problem is NP-hard. 
\end{theorem}

\begin{proof}
Proof can be checked in~\ref{app:proof}
\end{proof}

\spara{Comparing with \cite{sozio2010community}.}
Note that computing {\SEG} is the same with finding a solution to a community search problem~\cite{sozio2010community}. The problem is that : given a graph $G$ and \textcolor{black}{a} set of query nodes $Q$, their model aims to find a connected subgraph while maximizing \textcolor{black}{the} minimum degree such that 1) contains all the query nodes $Q$ and 2) all nodes in the subgraph are at a distance to $Q$ less than a threshold. Hence, by restricting $|Q|=1$ and applying a degree constraint, \cite{sozio2010community} can be utilized to find the {\SEG}. However, note that finding {\SEG}s is different  from finding a solution for {\LUEM}. When we use a greedy approach to find a solution for {\LUEM} by utilizing \cite{sozio2010community}, the approach is the same as our basic algorithm~\ref{sec:ba}.

\section{Approximation Algorithms}\label{sec:approx}

In this section, we introduce approximation algorithms to solve the {\LUEM} problem. 
First, we present a framework of the approximation algorithm using the characteristics of our objective function and present a basic algorithm(\BA). 
To improve efficiency, we present an effective $r$-neighbor-based algorithm (\ERA) with a pruning strategy.

\subsection{Computing {\SEG}}

\begin{algorithm}[ht]
\SetKw{return}{return}
\SetKw{null}{null}
\SetKwData{belongings}{belongings}
\SetKwFunction{kcore}{$k$-core}
\SetKwFunction{hEgo}{$r$-ego}
\SetKwFunction{connectedComp}{connectedComp}
\SetKwInOut{Input}{input}
\SetKwInOut{Output}{output}
\Input{$G=(V,E)$, user parameters $k$, $r$, seed node $s$}
\Output{{\SEG} $\mathcal{K}_{k,r}(s)$}
$\mathcal{K}_{k,r}(s) \leftarrow \varnothing$\;
$D \leftarrow$ \kcore{$G[N_r(s)]$, $k$} \;
\If{$D$ contains $s$}{
    \return \connectedComp{$D$, $s$}\;
}
\return $\null$\;
\caption{Computing \textsf{SEG}}
\label{alg:SEG_compute}
\end{algorithm} 

In this section, we present how to compute an {\SEG} given a seed node $s$ using Algorithm~\ref{alg:SEG_compute}. 
Computing an {\SEG} is simple and intuitive. First, we obtain an induced network of {\rnbr} from  the seed node $s$, then compute the $k$-core (line 2). Because the $k$-core returns multiple connected components, we select a connected component that contains the seed node $s$ (line 4). If there is no connected component that contains  seed node $s$ in the $k$-core, the algorithm returns null (line 5). Hence, computing an {\SEG} requires $O(|V|+|E|)$, and computing all  {\SEG}s in a graph requires $O(|V|(|V|+|E|))$.

\begin{figure}[t]
\centering
\includegraphics[width=0.99\linewidth]{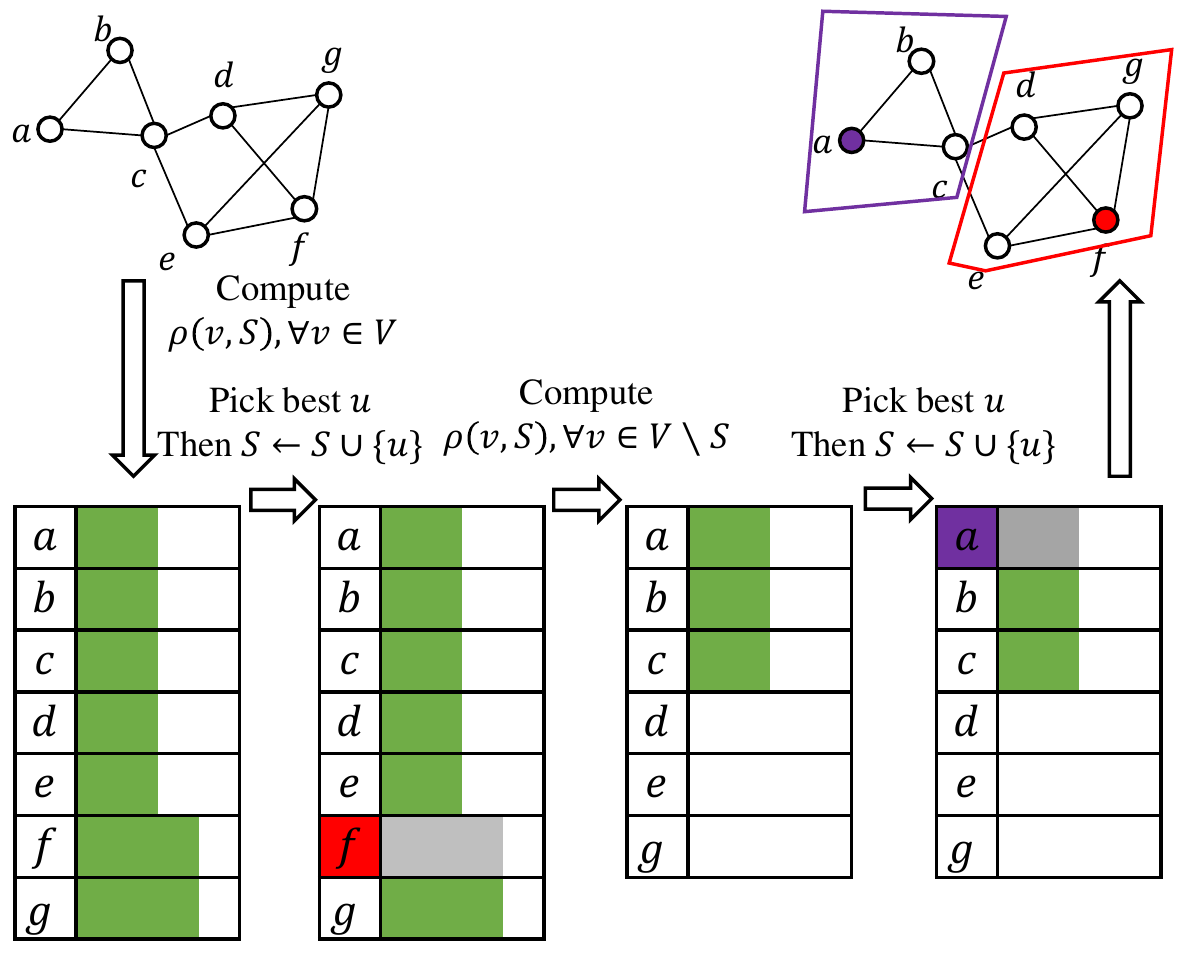}
\caption{Basic algorithm ($b=2, k=2, r=1$)}
\label{fig:baseline_example} 
\end{figure}

\subsection{Framework}

\begin{algorithm}[ht]
\SetKw{break}{break}
\SetKw{return}{return}
\SetKw{next}{next}
\SetKw{true}{true}
\SetKw{AND}{AND}
\SetKw{null}{null}
\SetKwData{C}{C}
\SetKwFunction{SEG}{SEG}
\SetKwInOut{Input}{input}
\SetKwInOut{Output}{output}
\Input{$G=(V,E)$, $k$, $r$, $b$}
\Output{A set of seed nodes $\mathcal{S}$}
$\mathcal{S} \leftarrow \varnothing$\;
\For{$|\mathcal{S}| \neq b$}{
    $s \leftarrow \argmax_{v\in V}$ $\rho(v, \mathcal{S})$\;
    $\mathcal{S} \leftarrow$
    \textcolor{black}{ $\{s\} \cup \mathcal{S}$}\;
}
\return $\mathcal{S}$\; 
\caption{\mbox{Algorithmic framework of hill-climbing approach}}
\label{alg:framework}
\end{algorithm}

Algorithm~\ref{alg:framework} depicts the framework of the $1-\frac{1}{e}$ approximation algorithm. This framework incorporates a widely used greedy optimization search named hill-climbing approach~\cite{nemhauser1978analysis} to find a solution.
It iteratively identifies a seed node that maximizes the number of engaged users when it merges with the current solution $\mathcal{S}$. Note that after adding a seed node, the number of engageable users of possible seed nodes can be decreased or unchanged. Therefore, it is required to check which seed node has the largest number of disengaged users at every iteration. After finding $b$ seed nodes, the procedure is terminated. 

\begin{algorithm}[ht]
    \SetKw{break}{break}
    \SetKw{return}{return}
    \SetKw{next}{next}
    \SetKw{true}{true}
    \SetKw{AND}{AND}
    \SetKw{null}{null}
    \SetKwData{C}{C}
    \SetKwFunction{SEG}{SEG}
    \SetKwFunction{kcore}{$k$-core}
    \SetKwFunction{updateEF}{updateEF}
    \SetKwInOut{Input}{input}
    \SetKwInOut{Output}{output}
    \Input{$G=(V,E)$, $k$, $r$, $b$}
    \Output{A set of seed nodes $\mathcal{S}$}
    $\mathcal{S} \leftarrow \varnothing$\;
    $D \leftarrow $ \kcore{$G$,$k$}\;
    $C \leftarrow \{v_1: \mathcal{K}_{k,r}(v_1), v_2: \mathcal{K}_{k,r}(v_2), \ldots, v_{|D|}: \mathcal{K}_{k,r}(v_{|D|}) \}$ \;
    \For{$|\mathcal{S}| \neq b$}{
        $s \leftarrow \argmax_{v \in C.\text{keys}} \rho(v, \mathcal{S})$ \tcp*{To compute $\rho$, $\mathcal{K}$ is utilized}
        $\mathcal{S} \leftarrow \{s\} \cup \mathcal{S}$\;
    }
    \return $\mathcal{S}$\; 
    \caption{Basic Algorithm}
    \label{alg:realImpl}
    \end{algorithm}

\subsection{Basic algorithm(\BA)}\label{sec:ba}
The direct implementation of the Algorithm~\ref{alg:framework} is the basic algorithm (\BA) with a $1-\frac{1}{e}$ approximation ratio. The procedure of the {\BA} is as follows. 

\begin{pDure}
\textcolor{black}{
At the initial stage, it computes all the {\SEG}s (Lines 1-3). Then, until finding $b$ seed nodes, it iteratively finds a node that can maximize the number of engaged users. The selected node will be added to the solution (Lines 4-6). Finally, it returns the selected seed nodes as a result (Line 7). 
}
\end{pDure}

\begin{figure*}[ht]
\centering
\includegraphics[width=0.7\linewidth]{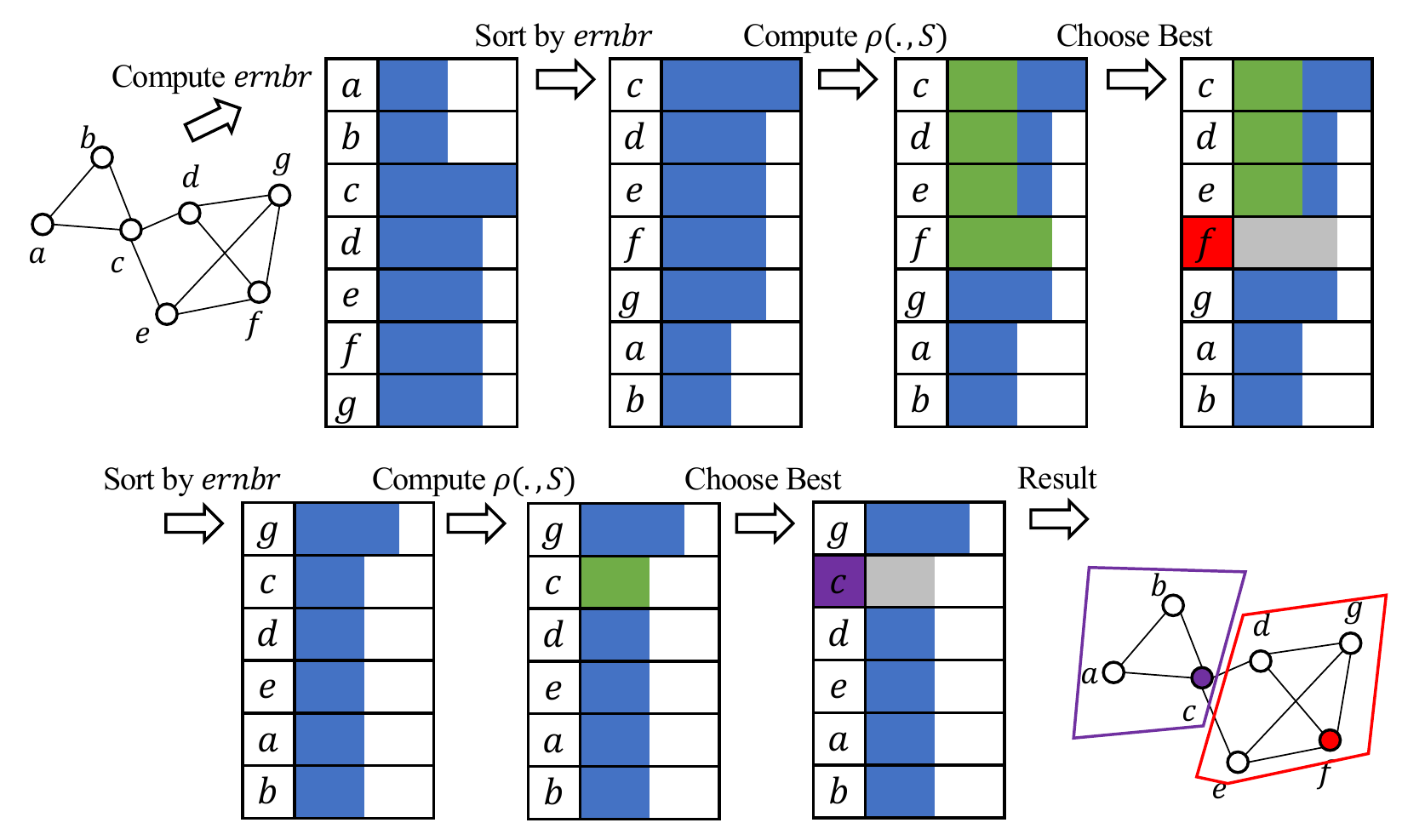}
\caption{{\ernbr}-based pruning ($b=2, k=2, r=1$)}
\label{fig:ernbr_example}
\end{figure*}

\begin{example}
In Figure~\ref{fig:baseline_example}, suppose that $k=2, r=1$, and $b=2$. To find a solution, we first select a node {\textquotesingle}$f${\textquotesingle} because the $\rho(f, S)$ is $4$. Therefore, the nodes $\{d, e, f, g\}$ are engaged. 
Next, we select a node {\textquotesingle}$a${\textquotesingle}, which enables engaging nodes $\{a, b, c\}$. When we find two seed nodes $\{f, a\}$, we terminate the algorithm because $b=2$.
\end{example}
In the following, we briefly demonstrate the $1-1/e$ approximation ratio of the framework. Please refer to \cite{nemhauser1978analysis, kempe2003maximizing} for further details on this topic.

\spara{Theoretical analysis.} We next check the approximability of our proposed algorithm.  

\begin{lemma}\label{lemma:ineq}

Let $OPT$ denote an optimal solution. We suppose that there is an identified solutions $\mathcal{S}$ such that $|\mathcal{S}| < b$. Then the following inequality always holds.

\textcolor{black}{
\begin{align}\label{Equation:lemma1}
\max_{x\in V} [\rho(\mathcal{S}\cup \{x\}) - \rho(\mathcal{S})]\geq \frac{1}{b} [\rho(OPT) - \rho(\mathcal{S})]
\end{align}
}

\end{lemma}

\begin{proof}
Proof of Lemma~\ref{lemma:ineq} can be checked in \ref{app:lemma_proof}
\end{proof}

\begin{theorem}\label{theorem:ratio}
Algorithm~\ref{alg:framework} holds $1-\frac{1}{e}$ approximation ratio, i.e., $\rho(\mathcal{S})\geq (1-\frac{1}{e}) \rho(OPT)$. 
\end{theorem}

\begin{proof}
Proof can be checked in \ref{app:ratio_proof}
\end{proof}

\spara{Time complexity.} The basic algorithm takes $O(|V|(|V|+|E|)+b|V|^2)$. The time complexity of each component is as follows. 
\begin{itemize}[leftmargin=*]
\item For initialization, it takes $O(|V|(|V|+|E|))$ to compute all {\SEG}s. 
    \item For each iteration, it is required to compute the set difference $|V|$ times. Since the set difference takes $O(|V|)$ time complexity, it takes $O(b|V|^2)$. 
\end{itemize}

\spara{Limitation of {\BA}.}
We point out three major limitations of {\BA}: (1) \textit{Memory consumption} : it requires considerable memory to store all  {\SEG}s. It takes $O(|V|^2)$ memory space in the worst case. (2) \textit{Initialization bottleneck} : it requires computing all  {\SEG}s in the initial stage, which requires $O(|V|(|V|+|E|))$ time complexity. (3) \textit{Update bottleneck} : at every iteration, we need to check all {\SEG}s to find the best seed node.

In the following sections, we propose an efficient method to address \textcolor{black}{the above} three issues. Specifically, {\ERA} focuses on resolving the update bottleneck issue of the {\BA} while preserving  \textcolor{black}{its} effectiveness. In section~\ref{sec:fast}, we solve these issues simultaneously without
losing accuracy.

\subsection{Effective {\rnbr}-based approximation algorithm}

In section~\ref{sec:ba}, we discuss the characteristics of the $\rho$ function and approximation ratio of {\BA}. 
Even if {\BA} is effective (a.k.a. $1-1/e$), it cannot sufficiently handle large datasets. 
Therefore, in this section, we propose a pruning technique called the {\rnbr}-based pruning strategy to improve efficiency (update bottleneck) while preserving  effectiveness. First, we define some terminology. 

\begin{definition}
(\underline{Effective $r$-neighbor size(\ernbr)}). \\ 
\textcolor{black}{
Given a graph $G=(V,E)$, a set of seed nodes $\mathcal{S}^i$ at current iteration $i$, node $v\in V$, minimum degree threshold $k$, and distance threshold $r$, effective $r$ neighbor size ({\ernbr}) of $v$, denoted as $\mathcal{E}_{r}(v)$, is defined as follows. 
%
\begin{align}
    \mathcal{E}_{r}(v) = 
        \begin{cases}  
        |N_{r}(v)| & \text{default value} \\
        \rho(v, S^j) & \text{if } \rho(v, S^j) \text{ where $j< i$ is already computed}
        \end{cases}
\end{align}
}
\end{definition}

\textcolor{black}{
We can consider that {\ernbr} implies the number of possible engageable nodes based on our observation. This {\ernbr} is helpful to prune a set of nodes as an upper bound.
}
For example, at the initial stage, all  $|N_{r}(v)|$  is the same with {\ernbr} because all nodes in $N_r(v)$ can be engaged \textcolor{black}{(without considering the degree constraint)}. When we compute $\rho(v, S)$, we update the {\ernbr} value of node $v$ to keep the recent value. 

Note that at the initial stage, we do not need to compute {\SEG}s of all nodes because it takes a long time.  We only compute the {\SEG} if it is required to be computed; in other words, we adopt the lazy update manner. Note that $|\mathcal{E}_{r}(v)|$ is always larger than or equal to $\rho(v, S)$, that is, it can be an upper bound of $\rho(v, S)$. Thus, we can design the following pruning strategy.

\begin{algorithm}[ht]
\SetKw{break}{break}
\SetKw{return}{return}
\SetKw{next}{next}
\SetKw{true}{true}
\SetKw{AND}{AND}
\SetKw{null}{null}
\SetKwData{C}{C}
\SetKwData{U}{U}
\SetKwFunction{SEG}{SEG}
\SetKwFunction{sortBySize}{sortByValue}
\SetKwFunction{ernbra}{ernbr}
\SetKwInOut{Input}{input}
\SetKwInOut{Output}{output}
\Input{$G=(V,E)$, $k$, $r$, $b$}
\Output{A set of seed nodes $\mathcal{S}$}
$\mathcal{S} \leftarrow \varnothing$\;
$H \leftarrow \{\{v_1, |N_r(v_1)|\}, \{v_2, |N_r(v_2)|\}, \ldots, \{v_{|V|},|N_r(v_{|V|}|)\} \} $\;
$H \leftarrow $ \sortBySize{$H$}\tcp*{Sort by {\ernbr}}
\For{$|\mathcal{S}| \neq b$}{
    $Cur, \U \leftarrow \varnothing$\;
    \For{$u\in H.keys()$}{
        $\mathcal{E}_r(u) \leftarrow H.get(u)$ \;
        \If{$\rho(Cur, \mathcal{S}) \geq \mathcal{E}_r(u)$}{
           \break\;
        }
        \If{$\rho(Cur, \mathcal{S})$ $<$ $\rho(u, \mathcal{S}$)}{
           $Cur \leftarrow u$\;
        }
        $\mathcal{E}_r(u) \leftarrow$ $\rho(u, \mathcal{S})$\;
        $U \leftarrow U \cup \{u\}$\;
        
    }
    $\mathcal{K} \leftarrow \mathcal{K} \cup Cur$\;
    $H \leftarrow H\setminus Cur$\;
    Update order $U \setminus \{Cur\}$ of $H$ \tcp*{using binary search}
}
\return $\mathcal{K}$\; 
\caption{Pseudo description of {\ERA}} 
\label{alg:erpruning}
\end{algorithm}

Instead of computing all  {\SEG}s at the beginning of the algorithm, an effective {\rnbr}-based pruning iteratively computes a few {\SEG}s that are promising candidates for selection. The high-level idea is to avoid computing {\SEG}s using {\ernbr} because {\ernbr} of a specific node $v$ is an upper bound of  $\rho(v,S)$. 
The detailed procedure \textcolor{black}{is} described as follows: 

\begin{pDure}\label{pdure:pruning_ernbr}
\textcolor{black}{
At the initial stage, the nodes are ordered based on {\ernbr} in descending order (Lines 1-3), then, we compute {\SEG}s to iteratively obtain the  $\rho$ value (Lines 8-11). 
Note that we retain the current best node $Cur$, which has the largest $\rho$ value (Line 11).
As a result of checking the nodes iteratively, if the {\ernbr} of the current node is less than or equal to $\rho$ of $Cur$, we return $Cur$ as the selected seed node of our algorithm (Lines 8-9). After selecting the best node, we update the node order based on the {\ernbr} (Line 16). This process is repeated until we identify $b$ seed nodes (Lines 4-16). Finally, it returns a set of selected seed nodes as a result (Line 17).  
}
\end{pDure}

\begin{example}
Figure~\ref{fig:ernbr_example} depicts the procedure~\ref{pdure:pruning_ernbr}. The sample graph consists of seven nodes and ten edges. 
We set $b=2$, $k=2$, and $r=1$. First, we initialize {\ernbr}. Next, we sort the nodes according to  {\ernbr}. Then, we iteratively compute {\SEG}s in a descending order of {\ernbr}. In the case of node $c$, {\SEG} is $\{a, b, c\}$. 
When we compute {\SEG} of node $f$, it contains $\{d, e, f, g\}$. Because {\ernbr} of the next node $g$ is $4$, notice that all  $\rho$ values of the nodes to be computed later are less than $4$. This implies that we do not need to compute {\SEG}s of the other nodes. Hence, we select the {\SEG} of node $f$ as a solution and delete  node $f$ from the candidate list. 
Then, we sort the list based on the size of {\ernbr}. We repeatedly compute {\SEG}s. Next, we select the {\SEG} of node $c$. Because the solution size is exactly the same as the number of seed nodes $b$, we terminate the algorithm. 

Note that if the inverted index is maintained to retain the information regarding which {\SEG} contains a specific node $u$, we can improve the efficiency. However, This is not preferred, because it requires considerable memory space. 
\end{example}

Before we compute an {\SEG}, we must check whether the {\SEG} has already been computed or not. If computed, we do not need to recompute it because the {\SEG} of a node is unique.

\spara{Comparing with \cite{minoux1978accelerated}.} 
In \cite{minoux1978accelerated}, authors propose an accelerated greedy algorithm (\AGA) to find a solution when the objective function is submodular. It iteratively maintains the $\Delta(v)$ value by computing $\rho(v)-\rho(\emptyset), \forall v\in V$ then updates the $\Delta$ value. Note that the high-level idea of {\AGA} and {\ERA} is similar. The two major differences compared with {\AGA} is as follows: 
(1) {\ERA} does not directly compute the value $\rho(v)$ since computing the {\SEG} is time-consuming. Thus, we utilize $|N_r(v)|$ to improve the efficiency since $|N_r(v)|$ can be utilized as an upper bound; 
(2) For every iteration, {\SEG} utilize binary search to maintain the sorted values. Thus, it can improve the efficiency to find the node which has the largest gain.

\spara{Time complexity.} The time complexity of {\ERA} is the same as that of {\BA} since it needs to compute all {\SEG}s in the worst case. However, as discussed in section~\ref{sec:experiment}, we observe that {\ERA} is much faster than {\BA} in practice.

\begin{figure}[t]
\centering
\includegraphics[width=0.95\linewidth]{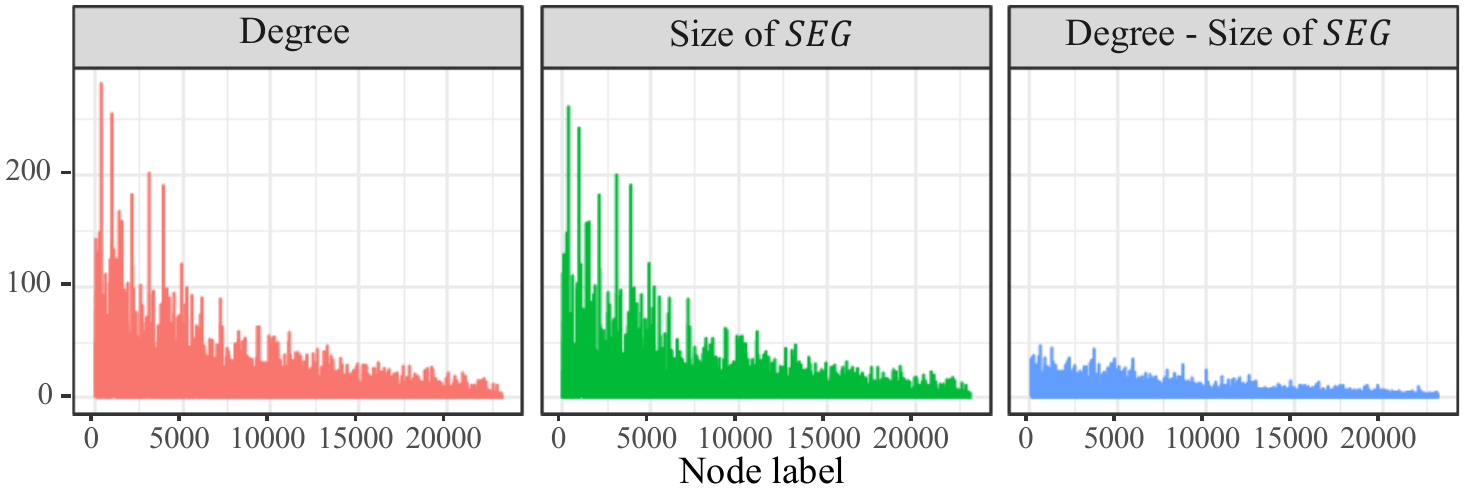}
\caption{Degree and the size of \textsf{SEG} in Condmat dataset}
\label{fig:cond_deg_krcc}
\end{figure}

\section{Fast Circle Algorithm}\label{sec:fast}
In section~\ref{sec:approx}, we discuss two algorithms for finding an approximate solution to the {\LUEM} problem. 
Despite improving the update bottleneck of {\BA}, {\ERA} intrinsically suffers from memory consumption and initialization bottlenecks.
Hence, in this section, we present an approach to improve the efficiency of {\ERA} by improving the abovementioned issues while preserving its effectiveness. 
Our {\FCA} algorithm is designed based on the assumption that a node with a large degree may engage many users.
In Figure~\ref{fig:cond_deg_krcc}, we present the degree and size of engageable users when $r=1$ and $k=3$ in the Condmat dataset~\cite{leskovec2007graph}.
Observe that the degree of a node is correlated with the {\SEG} size. In {\FCA}, we utilize this characteristic to find a solution for {\LUEM}. 

First, we introduce HyperANF~\cite{boldi2011hyperanf}, which is a technique used to approximate the neighborhood function. Next, we propose an efficient algorithm that incorporates HyperANF, called the Fast Circle Algorithm (\FCA). 

\begin{figure}[t]
\centering
\includegraphics[width=0.85\linewidth]{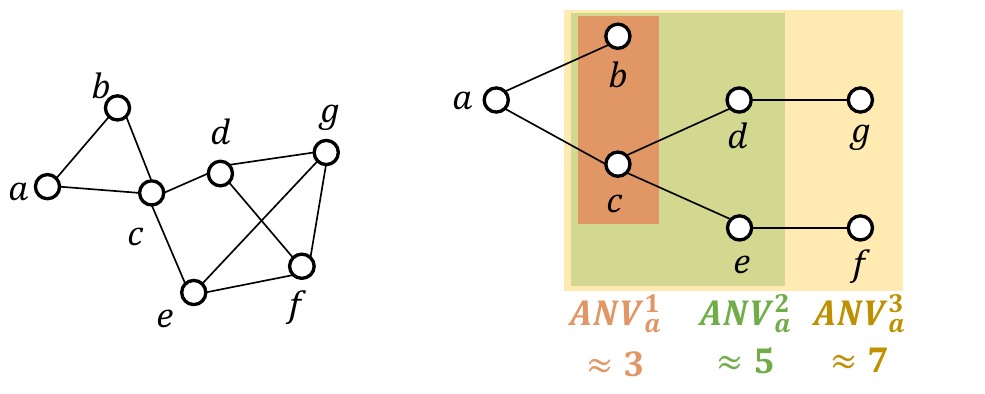}
\caption{Example of $ANV$}
\label{fig:anv_example}
\end{figure}

\begin{figure*}[t]
\centering
\includegraphics[width=0.7\linewidth]{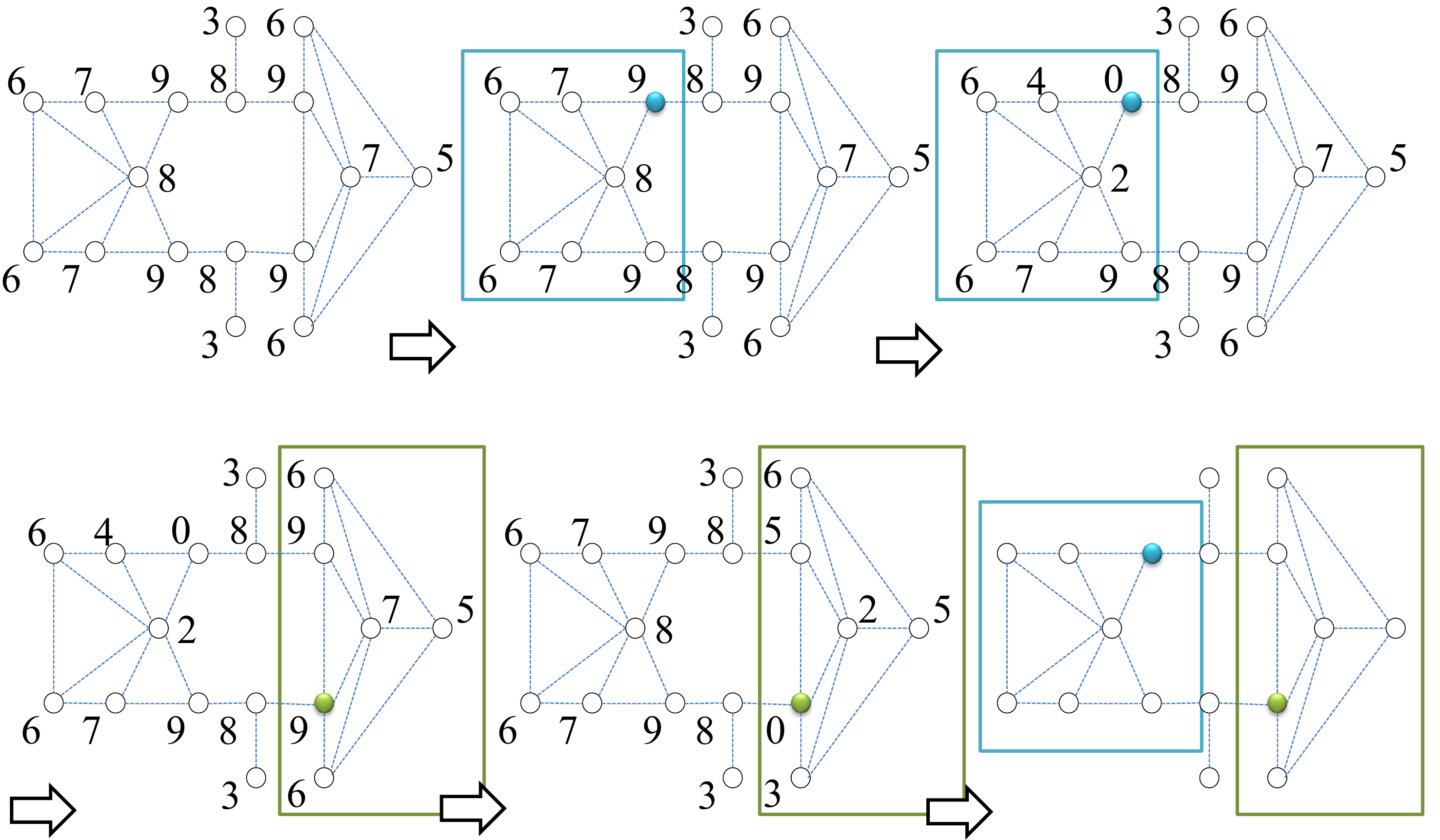}
\caption{Example of {\FCA} ($b=2$, $k=2$, $r=2$)}
\label{fig:fca_example}
\end{figure*}

\subsection{HyperANF}\label{sec:hyper}
HyperANF~\cite{boldi2011hyperanf} is a state-of-the-art algorithm to compute an approximation of the neighborhood function of a graph. 
It utilizes HyperLogLog counters~\cite{flajolet2007hyperloglog} which are statistical counters requiring $O(\log \log {n})$ bits. 
Thus, HyperANF can approximate the number of reachable nodes within a specific distance from a node. 
Unfortunately, the time complexity of the HyperANF is unknown but its running time  is expected to be approximately $O((|V|+|E|)h)$, where $h$ denotes the maximum distance. This is because  HyperANF is an extension of ANF~\cite{palmer2002anf} that requires $O((|V|+|E|)h)$. 

\subsection{Algorithm description}

By utilizing HyperANF, each node has a set of approximated neighborhood values ($ANV$s) for each distance, that is, given a radius $r$, each node has $1$-$ANV$, $2$-$ANV$, $\cdots$, $r$-$ANV$. First, we formally define the approximated neighborhood value ($ANV$) of a node. 

\begin{definition}
(\underline{Approximated neighborhood value ($ANV$)}). \\ Given a graph $G$, node $v\in V$, and threshold $r$, the approximated neighborhood value ($ANV$) is a key-value structure, in which a key is a set of integers from $1$ to $r$ and a value is an approximated number of neighbors from node $v$ within a specific distance. 
\end{definition}

\begin{example}
We illustrate the ANF value by reusing Figure~\ref{fig:twoKrcc}. We assume that $r=3$ and compute the $ANV$ of node $a$. Figure~\ref{fig:anv_example} shows the $ANF$ result of the node $a$. Note that $ANV$ does not return the exact degree because it is an approximated value. 
\end{example}

The intuition behind {\FCA} is as follows. We focus only on $ANV_v^{r}$ to select the best node. After finding node $u$ in the first iteration, we update the $ANV^{r}$ values of all nodes. If node $x$ is the $d$-hop neighbor of the selected node $u$, $ANV_{x}^{r}$ must be updated by negating $ANV_{x}^{r-d}$ because we assume that $ANV_{x}^{r-d}$ nodes are already engaged.

        \begin{algorithm}[ht]
        \SetKw{break}{break}
        \SetKw{return}{return}
        \SetKw{next}{next}
        \SetKw{true}{true}
        \SetKw{AND}{AND}
        \SetKwData{C}{C}
        \SetKwData{U}{U}
        \SetKwFunction{krcc}{krcc}
        \SetKwFunction{sortBySize}{sortByANV}
        \SetKwFunction{ef}{ef}
        \SetKwFunction{ernbra}{ernbr}
        \SetKwFunction{remove}{remove}
        \SetKwFunction{nextr}{next}
        \SetKwFunction{add}{add}
        \SetKwFunction{kcore}{kcore}
        \SetKwFunction{maxDistance}{maxDistance}
        \SetKwFunction{getDist}{getDist}
        \SetKwInOut{Input}{input}
        \SetKwInOut{Output}{output}
        \Input{$G=(V,E)$, $k$, $r$, $b$}
        \Output{A set of seed nodes $\mathcal{S}$}
        $D \leftarrow G[$\kcore{$G$, $k$}$]$, $\mathcal{S} \leftarrow \varnothing$\;
        initialize HyperANF on $D$\;
        \For{$r' \leftarrow 1:r$}{
            Setup $ANV^{r'}$  by computing HyperANF.\nextr{$r'$}\;
        }
        $D \leftarrow$ \sortBySize{$D$}\;
        \While{$|\mathcal{S}| \neq b$}{
            $\U \leftarrow \varnothing$\;
            $u \leftarrow \argmax_{u\in D}$ $ANV_u^{r}$\;
            compute $\mathcal{K}_{k,r}(u)$\;
            compute shortest path tree $T(u)$ using $G[\mathcal{K}_{k,r}(u)]$ with root $u$\;
            $d_T \leftarrow $ \maxDistance{$T(u)$}\;
            \For{$y \in \mathcal{K}_{k,r}(u)$}{
                Update the ANV values by Strategy\;
                \U.\add{$y$}\;
            }
            $D$.\remove{$u$}\;
            $\mathcal{S} \leftarrow$  \textcolor{black}{$\mathcal{S} \cup \{u\}$}\;
            update nodes $\U$ in $D$\;
        }
        \return $\mathcal{S}$\;
        \caption{Procedure of \textsf{FCA}} 
        \label{alg:FCA}
        \end{algorithm}

\begin{pDure}\label{pDure:FCA}
\textcolor{black}{
At the initial stage, given a distance threshold $r$, a set of $ANV$s ($ANV^1$ $\cdots$ $ANV^r$) for every node based on HyperANF is computed (Lines 1-4). Then, node $u$ that has the largest $ANV_u^r$ value is selected (Line 8). Intuitively, node $u$ have many $r$-neighbor nodes. 
Next, we compute {\SEG} $\mathcal{K}(u)$ and the shortest distances from node $u$; then, construct a shortest-path tree based on the shortest path (Lines 8-10). 
Next, if node $v$ is in the $i$-th level of the shortest-path tree of $\mathcal{K}(u)$ owing to rooting node $u$, we obtain $ANV_v^{r-i}$. Then, the $ANV$ values are updated as follows (Lines 11-14). 
\begin{itemize}
 \item $ANV_v^{r'} = ANV_v^{r'} - ANV_v^{r-i}, \forall r\geq r' > r-i$
 \item $ANV_v^{r''} = 0, \forall r'' \leq r-i$ 
\end{itemize}
This process is repeated for all nodes in $\mathcal{K}_{k,r}(u)$. After  \textcolor{black}{updating} $ANV$ value of a node, we relocate the position of the node to preserve the order in $D$ (Line 17). This procedure is repeated until the size of the solution is $b$.
}
\end{pDure}

\begin{example}
In Figure~\ref{fig:fca_example}, we give an example of {\FCA}. Suppose that $b=2$, $r=2$, and $k=2$. 
At the initial stage, we compute $ANV$s of all nodes. In the Figure, the numbers indicate the $ANV^r$ values of the nodes. 
Next, we choose the node that has the largest $ANV^r$ value. We choose a blue-coloured node in the first iteration. After choosing the node, we update the $ANV^r$ values of its neighbor nodes based on the procedure~\ref{pDure:FCA}. Among the updated $ANV^r$ values, we choose the green-coloured node that has the largest $ANV^r$ value. Then, we choose two blue and green nodes as the seed nodes. Consequently, the algorithm is terminated. Note that thirteen users are completely engaged. 
\end{example}

\spara{Time complexity.} {\FCA} takes $O(H + br|V| + b(|V| + |E|\log{|V|}))$. Time complexity of each component is as follows. 
\begin{itemize}[leftmargin=*]
 \item Computing $b$ {\SEG} takes $O(b(|V|+|E|))$
 \item Computing HyperANF\footnote{Note that the time complexity of HyperANF is unclear (See section~\ref{sec:hyper})} takes $O(H)$
 \item Updating $ANV^r$ takes $O(br|V|)$ 
 \item Computing the shortest path $O(b(|V| + |E|\log{|V|}))$
\end{itemize}

\section{Experiments}\label{sec:experiment}

We evaluated the proposed algorithms using several real-world networks. All experiments were conducted on Ubuntu 14.04 with 64GB memory and 2.50GHz Xeon CPU E5-4627 v4. For the implementation, we used the JgraphT library~\cite{jgrapht} and WebGraph~\cite{boldi2004webgraph}. Our code is publicly available~\footnote{\url{https://bit.ly/3GyC8wl} }. 
\textcolor{black}{
Since the dataset\textcolor{black}{s} are publicly available, we do not have any preprocessing. 
}

\spara{Dataset.} Table~\ref{tab:dataset} lists the basic statistics of real-world datasets. All datasets are publicly available.
$c^{max}$ and $d^{avg}$ denote the maximum core index and average degree, respectively.

\begin{table}[t]
\caption{Summary of the real-world datasets}
\label{tab:dataset}
\centering
\begin{tabular}{c|c|c|c|c}
\hline
 Name       & \textbf{\# nodes}     & \textbf{\# edges}   & \textbf{$c^{max}$} & \textbf{$d^{avg}$}      \\ \hline \hline
 Amazon\cite{yang2015defining}     & 334,863   & 925,872    & 6 & 5.52   \\ \hline
 Brightkite\cite{cho2011friendship} & 58,228    & 214,078   & 52  & 7.35      \\ \hline
 CondMat\cite{leskovec2007graph} & 23,133    & 93,497    & 25  & 8.08      \\ \hline
 DBLP\cite{yang2015defining}     & 317,080   & 1,049,866 & 113 & 6.62  \\ \hline
 Enron\cite{yang2015defining}     & 36,692   & 183,831 & 43 & 10.02   \\ \hline
 Hepth\cite{leskovec2007graph}   & 9,877     & 25,998    & 31  & 5.26     \\ \hline
 LA\cite{bao1, bao2}             & 500,597   & 1,462,501 & 120 & 5.84    \\ \hline
 NYC\cite{bao1, bao2}        & 715,605   & 2,552,603 & 157 & 7.13    \\ \hline
 Orkut\cite{yang2015defining}        & 3,072,441   & 117,185,083 & 253 & 76.28   \\ \hline
 Youtube\cite{yang2015defining}  & 1,134,890 & 2,987,624 & 51  & 5.27  \\ 
\hline \hline
\end{tabular}
\end{table}


\spara{\bf Algorithms.} As baseline algorithms, we used two graph centrality-based approaches\footnote{In Table~\ref{tab:centrality}, we show several centrality measures and notice that AC and CC \textcolor{black}{are} relatively scalable to handle large-sized datasets. Hence, we choose both centrality measures as the baseline algorithms.}, named AC(alpha centrality), and CC(clustering-coefficient). 
Both approaches compute the centrality measures; then, iteratively pick the node that has the largest centrality until $b$ seed nodes are identified. If {\SEG} of a selected node is null, the node is not selected. Owing to limited scalability, we only report the results of {\BA} in the scalability test. 

\begin{itemize}[leftmargin=*]
    \item Basic algorithm (\BA)
    \item Effective $r$-neighbors-based approximation algorithm (\ERA)
    \item Fast Circle Algorithm (\FCA) 
    \item Alpha Centrality-based approach (AC)
    \item Clustering Coefficient-based approach (CC)
\end{itemize}

\textcolor{black}{
\spara{Parameter setting.} In our experiments, we vary the values of three parameters $k$, $r$, and $b$. In all the experiments, we use values $r \in [2,3]$ since the larger values are not interesting since when $r$ becomes large, all the nodes are reachable from the seed nodes. We fix the value $b=10$ since the larger values may cover all the possible nodes (See Figure~\ref{fig:ama_btest}). For the value $k$, it controlled the cohesiveness level. To the best of our knowledge, many previous works on minimum-degree based cohesive subgraph discovery~\cite{kim2020densely,barbieri2015efficient,fang2017effective,fang2016effective} take the minimum degree as an input parameter and did not shed light on strategies for setting its value. Intuitively, when a user selects a large value $k$, we expect the {\SEG} might be more cohesive and smaller. Thus, we consider the value $k$ an additional degree of freedom available to users to specify the cohesiveness, and users can compare results with different parameters.
}

\begin{figure}[t!]
  \centering
  \begin{subfigure}[t]{.99\linewidth}
    \centering
    \includegraphics[width=0.99\linewidth]{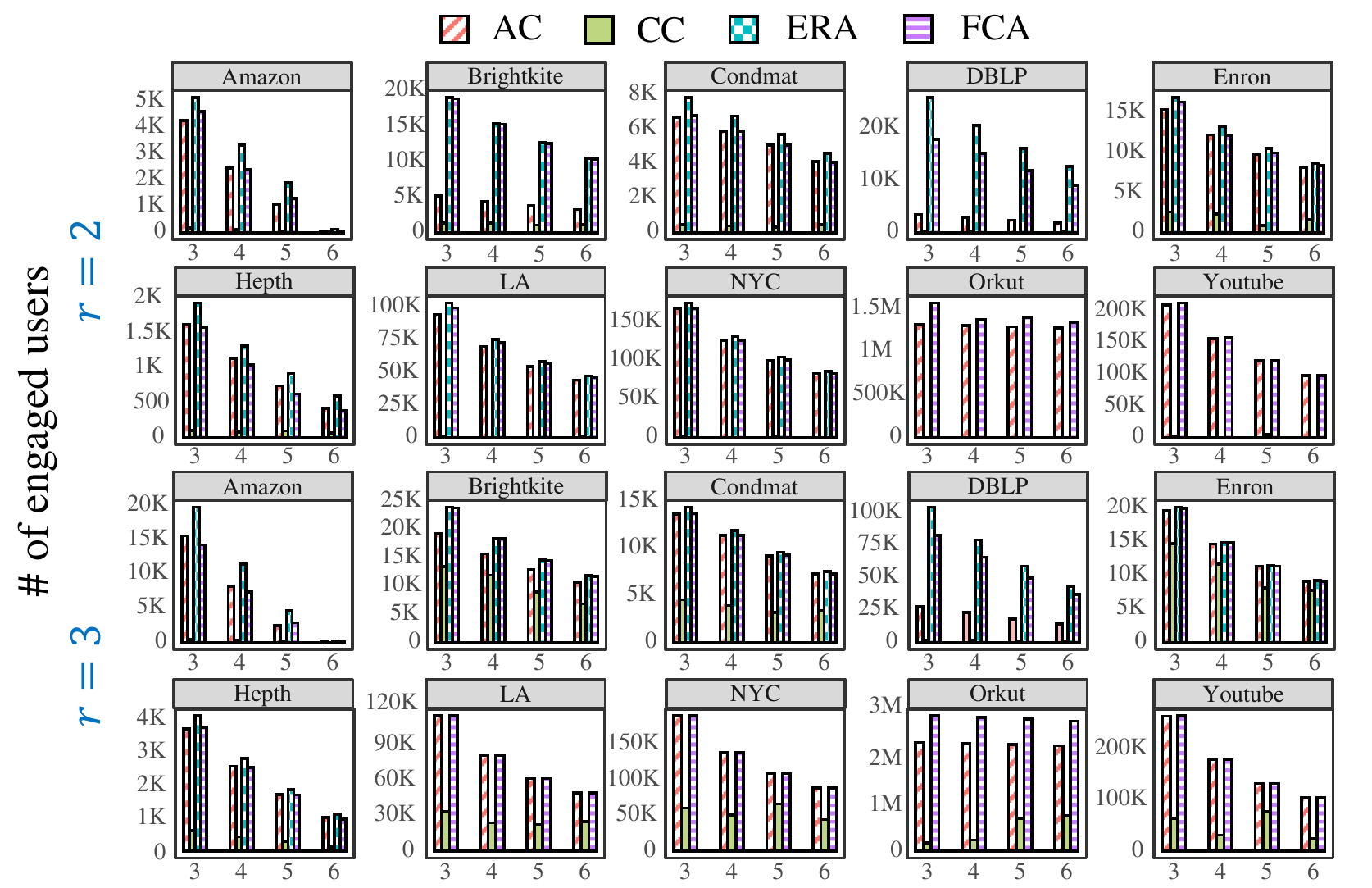}
    \vspace{-0.3cm}
    \caption{Effectiveness test}
    \label{fig:effectiveness_test}
  \end{subfigure}
  \begin{subfigure}[t]{.99\linewidth}
    \centering
\includegraphics[width=0.99\linewidth]{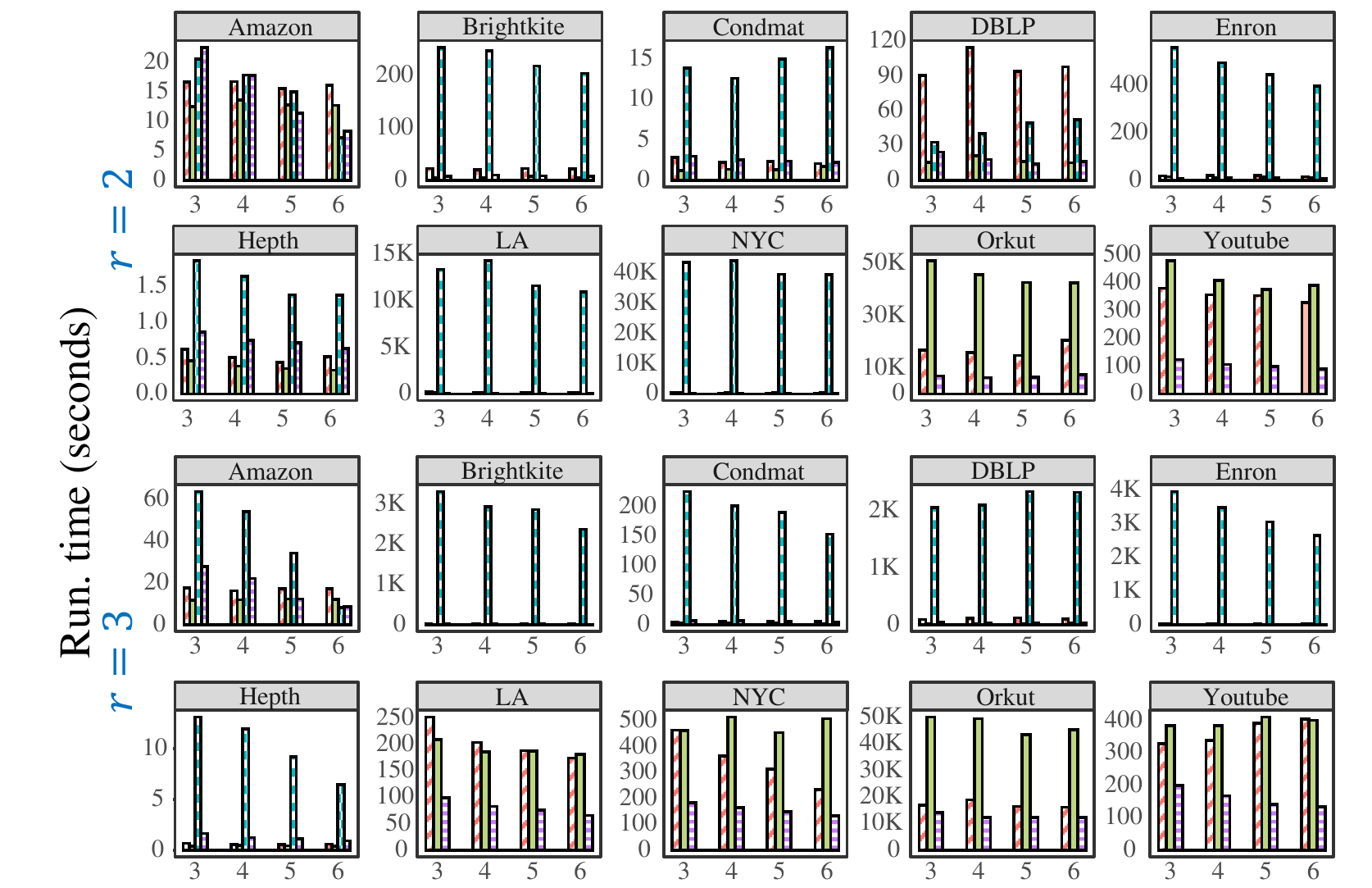}
\caption{Efficiency test (same legend with Figure~\ref{fig:effectiveness_test}) }
\label{fig:efficiency_test}
  \end{subfigure}
  \caption{Real-world networks}
  \label{fig:real-world_test}
\end{figure}

\spara{Effectiveness and efficiency evaluation in real-world networks.}
Figures~\ref{fig:effectiveness_test} and \ref{fig:efficiency_test} show effectiveness and efficient of the proposed algorithms on tests conducted using real-world datasets, respectively.
Each column indicates different datasets and each row indicates different $r$ values. Each Figure reports the number of engaged users and running time with respect to varying $k$ values. For all experiments, we set $b=10$.  
The effectiveness is verified in Figure~\ref{fig:effectiveness_test}. Observe that for all cases, {\ERA} algorithm returns the largest number of engaged users, and {\FCA} returns comparable effective results. The two baseline algorithms return relatively low-quality results.
In Figure~\ref{fig:efficiency_test}, we check the efficiency. 
Note that for large-sized datasets, such as LA, NYC, Orkut, and Youtube datasets with $r \geq 2$, {\ERA} algorithm does not finish within 24 hours. Therefore, we did not report these results. We observed that our {\FCA} outperformed the two baseline algorithms for all datasets. 
Next, we verified the efficiency of these algorithms. We observed that the proposed {\FCA} algorithm is much faster than the two baseline algorithms because it does not need to compute many {\SEG}s.

\spara{Measuring node influence.} In this experiment, we use the Polbooks network~\cite{nr} to verify the quality of our seed nodes  using node influence measures. 
The Polbooks network is a well-known network with $105$ nodes and $441$ edges. The nodes and edges represent blogs on U.S. politics and web links, respectively.
Because measuring the influence is time consuming, we used a small graph to check the tendency. 
In the experiments, we set $k=3$, $r=2$, and $b=3$. In addition, we used five measures to check the influence~\cite{salavaty2020integrated}: 
(1) ClusterRank~\cite{chen2013identifying} : It is a local ranking measure that considers not only the number of neighbor nodes and the neighbors' influences but also the clustering coefficient. 
(2) Hubness centrality~\cite{salavaty2020integrated} : It reflects the power of each node in its surrounding environment. 
(3) IVI score~\cite{salavaty2020integrated} :  IVI method is an integrative measure for determining influence nodes.  
(4) SIRIR~\cite{salavaty2020integrated} : It is an SIR-based influence ranking method that combines the leave-one-out cross-validation with a conventional susceptible-infected-recovered model.
(5) Spreading score~\cite{salavaty2020integrated} : It is indicative of the spreading potential.

\begin{table}[t]
\caption{Node influence}
\label{tab:influence}
\centering
\begin{tabular}{c|c|c|c}
\hline
\textit{Measures}  & Mean of all nodes & {\ERA} & {\FCA} \\ \hline \hline
clusterRank        & 48.0829           & \textbf{63.7308} & 63.7195 \\ \hline
hubness            & 28.5478           & \textbf{58.3913} & 56.9565 \\ \hline
IVI                & 12.6734           & 40.4316 & \textbf{60.693}  \\ \hline
SIRIR              & 1.4515            & 2.224 & \textbf{3.139}     \\ \hline
Spreading score    & 19.2765           & 48.6802 & \textbf{76.9068} \\ \hline \hline
\end{tabular}
\end{table}

Table~\ref{tab:influence} reports the experimental results. Observe that our identified seed nodes have larger influence scores than the average influence score of the nodes for all metrics. 
This indicates that the seed nodes play an important role in the graph structure. In addition, because {\FCA} returns high influence scores, large {\rnbr} is preferred when identifying influential nodes in networks.

\begin{table}[t]
\caption{Centrality score}
\label{tab:centrality}
\centering
\begin{tabular}{c|c|c|c}
\hline
\textbf{Algorithm} & Mean of all nodes & {\ERA}& {\FCA} \\ \hline \hline
AC & 1.096 & \textbf{1.1739} & 1.171 \\ \hline
BC & 104.2912 & 332.9455 & \textbf{488.5373} \\ \hline
CC & \textbf{0.4902} & 0.3479 & 0.3283 \\ \hline
Close & 0.3329 & 0.3716 & \textbf{0.407} \\ \hline
CN & 4.9702 & \textbf{6} & \textbf{6} \\ \hline
IM & 3.1193 & \textbf{5} & \textbf{5} \\ \hline
HC & 0.4015 & 0.4671 & \textbf{0.4888} \\ \hline \hline
\end{tabular}
\end{table}

\spara{Graph centrality.}
In this experiment, we reused the Polbooks network to verify the centrality score\footnote{Graph centrality is to measure the important nodes (or edges) in a graph.} of our seed nodes. 
Because our selected seed nodes might be located at the center of the core structure in a graph, we assume that the seed nodes might have a high centrality score. 
We checked the alpha centrality (AC)~\cite{bonacich2001eigenvector},  betweenness centrality(BC)~\cite{brandes2001faster}, clustering coefficient (CC), closeness centrality (Close)~\cite{bavelas1950communication}, coreness (CN) ~\cite{matula1983smallest}, influence maximization  (IM)~\cite{kempe2003maximizing}, and  harmonic centrality (HC)~\cite{newman2003structure,rochat2009closeness}.
Table~\ref{tab:centrality} reports the results of the centrality scores of our {\ERA} and {\FCA} algorithms and the average centrality scores of the nodes. Observe that {\ERA} has a larger value than the average centrality score. Moreover, CC only focuses on the local graph structure, and our seed nodes have relatively small scores.

\begin{table}[ht]
\caption{The number of engaged users}
\label{tab:engagedusers}
\centering
\begin{tabular}{c|c|c|c|c|c|c|c|c|c}
\hline
{\bf Algs}       & AC & BC & CC & Close & CN & HC & IM & {\ERA} & {\FCA} \\ \hline \hline
$\rho(.)$ & 56 & 78 & 31 & 73 & 49 & 56 & 26 & \textbf{103} & 102  \\ \hline \hline
\end{tabular}
\end{table}

In Table~\ref{tab:engagedusers}, we verify the number of engaged users when a set of nodes are selected based on centrality measures. Observe that our {\ERA} and {\FCA} algorithms outperform the other algorithms because they aim to maximize the number of engaged users. This indicates that maximizing the number of engaged users cannot be achieved by identifying classic node importance measures.

\begin{figure}[h]
\centering
\includegraphics[width=0.9\linewidth]{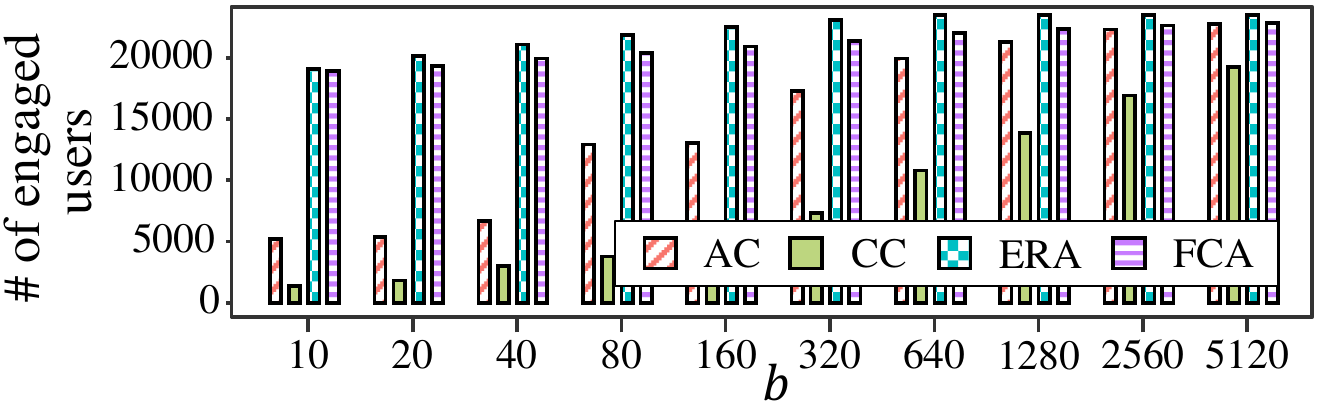}
\vspace{-0.2cm}
\caption{Varying parameter $b$}
\label{fig:ama_btest}
\end{figure}

\spara{Effect on $b$.}
To verify the effect on parameter $b$, we used the Brightkite dataset by setting $k=3, b=2$ and $r=2$, and reported the number of distinctly engaged users. 
In Figure~\ref{fig:ama_btest}, when $b$ value becomes relatively large, the difference between {\ERA} and {\FCA} algorithms increases. However, when $b$ value becomes very large (e.g., $5120$), the difference is negligible because  most nodes are already engaged.
Furthermore, we also observe that $AC$ and $CC$ return relatively few engaged users because the larger centrality scores do not indicate a larger number of engaged users.

\begin{figure}[h]
\centering
\includegraphics[width=0.8\linewidth]{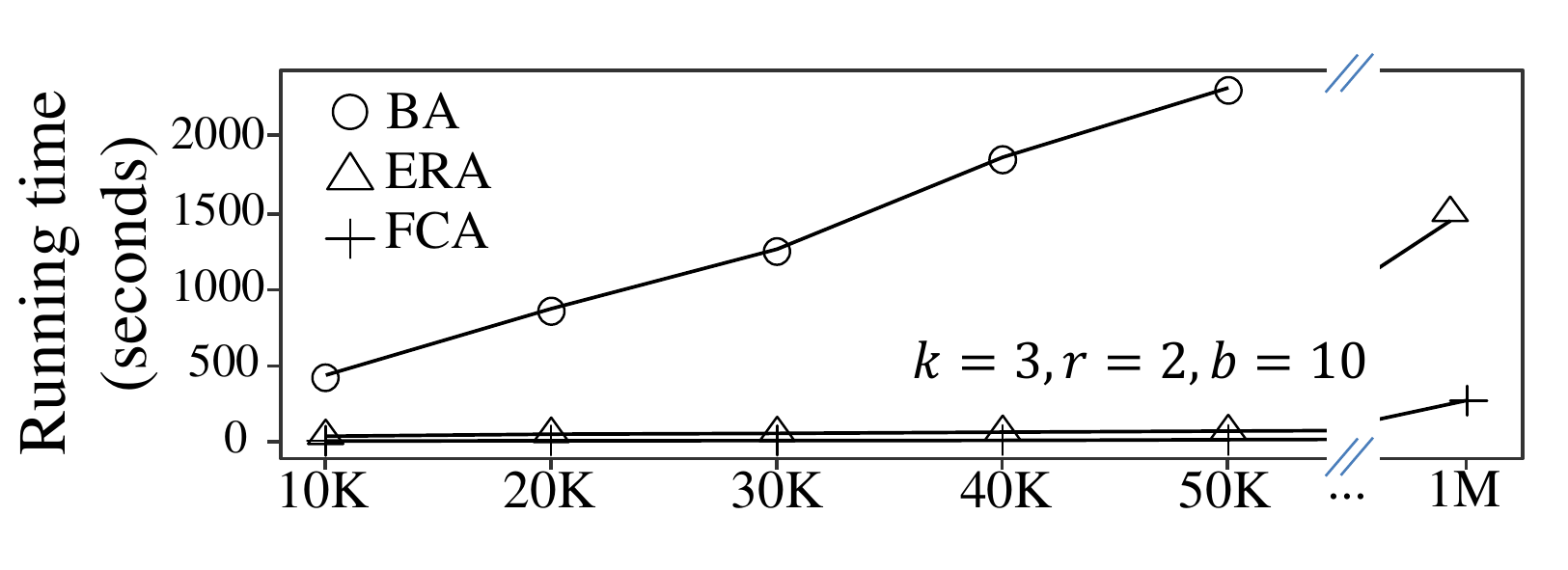}
\vspace{-0.2cm}
\caption{Scalability test on synthetic networks}
\label{fig:scalability}
\end{figure}

\begin{figure}
  \begin{subfigure}{.33\linewidth}
    \centering
    \includegraphics[width = \linewidth]{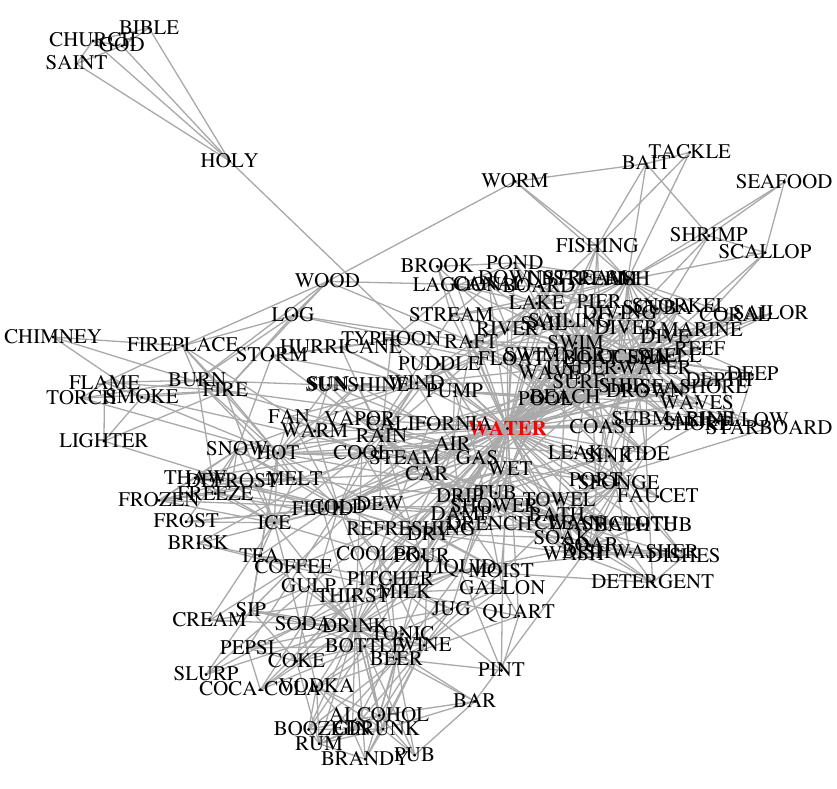}
    \caption{First {\SEG}}
  \end{subfigure}%
  \begin{subfigure}{.33\linewidth}
    \centering
    \includegraphics[width = \linewidth]{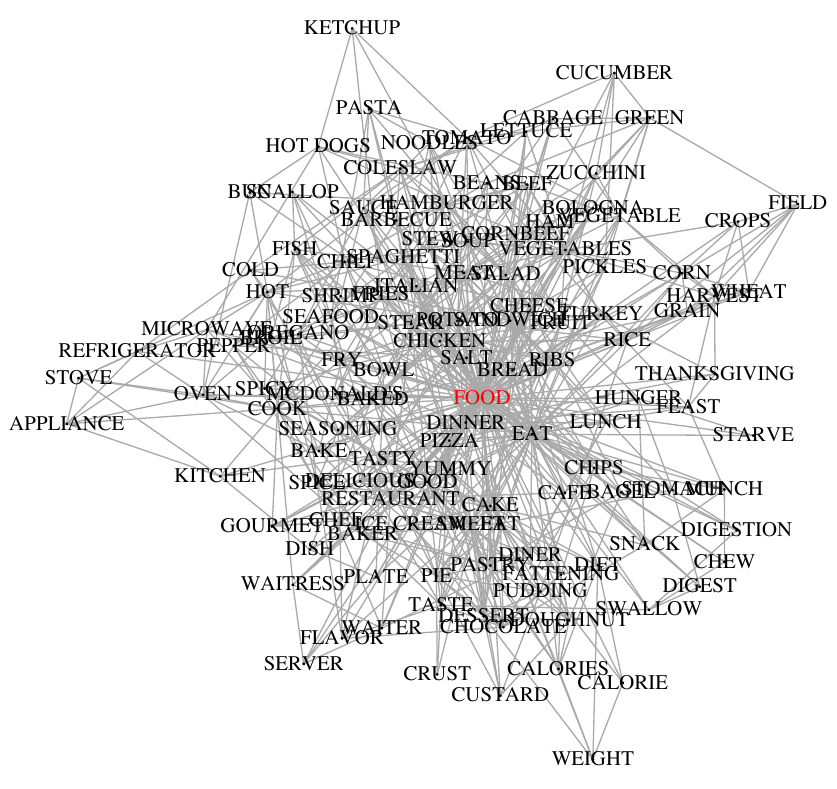}
    \caption{Second {\SEG}}
  \end{subfigure}%
  \begin{subfigure}{.33\linewidth}
    \centering
    \includegraphics[width = \linewidth]{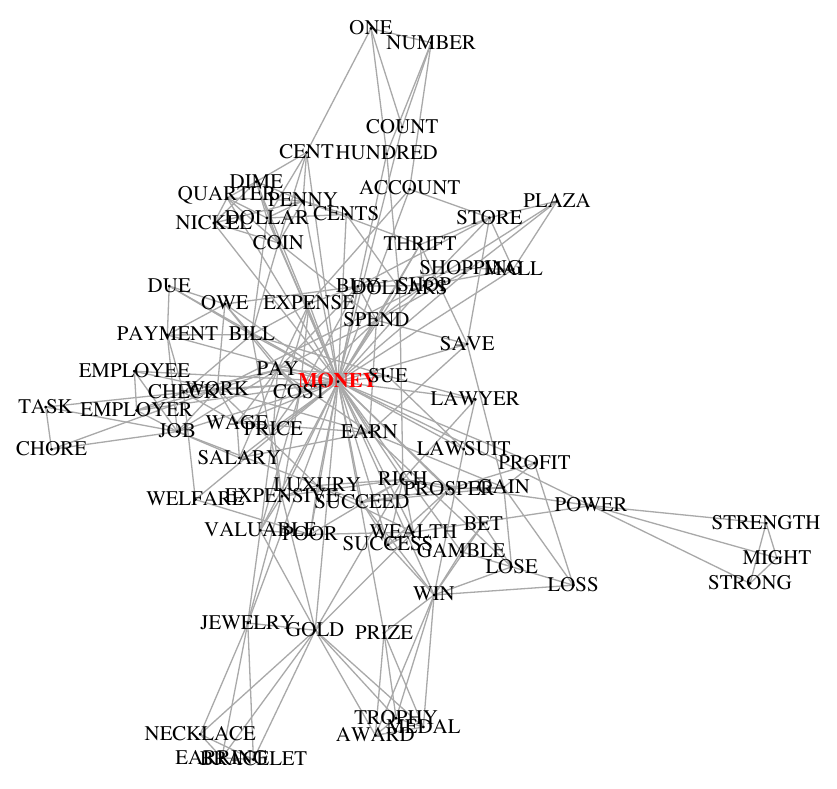}
    \caption{Third {\SEG}}
  \end{subfigure}
  \caption{Result of {\LUEM} in the word-association network ($k=3, r=2, b=3$)}
  \label{fig:case_word}
\end{figure}

\spara{Scalability test.} Figure~\ref{fig:scalability} demonstrate the efficiency of our algorithms with respect to \textcolor{black}{a} varying number of nodes in the LFR synthetic networks~\cite{lancichinetti2008benchmark} to present scalability. 
Notice that the running time of {\ERA} increases  significantly when the number of nodes increases. 
Moreover, observe that the running time of {\FCA} is much faster than that of {\ERA} and increases almost linearly with node size.


\spara{Case study : Word-association network}
Figure~\ref{fig:case_word} depicts the three resultant {\SEG}s in the word-association network~\cite{nelson2004university}. The word-association network consists of $10,617$ nodes and $72,168$ edges. Each node indicates an English word, and each edge indicates the association between two nodes. Please refer. To avoid meaningless association, we retain the edges if the edge weight (occurrence) is greater than or equal to $10$. By utilizing {\LUEM} with $b=3$, we identified three seed nodes: \textit{Water}, \textit{Food}, and \textit{Money}. We deduced that the three seed nodes are keywords that represent high-level concepts (topic) to explain the words. For example, in the first {\SEG}, many keywords such as Coast, Leak, Sink, Wet, Port, Gas, Vapor, Pump, and Puddle are related to the seed node \textit{Water}. In the second {\SEG}, we noticed that most of the keywords, such as Rib, Bread, Pizza, Pastry, Chips, and Spaghetti are related to the keyword \textit{Food}.
Notice that each {\SEG} is a cluster of similar keywords.

\begin{figure}[t]
\centering
\includegraphics[width=0.99\linewidth]{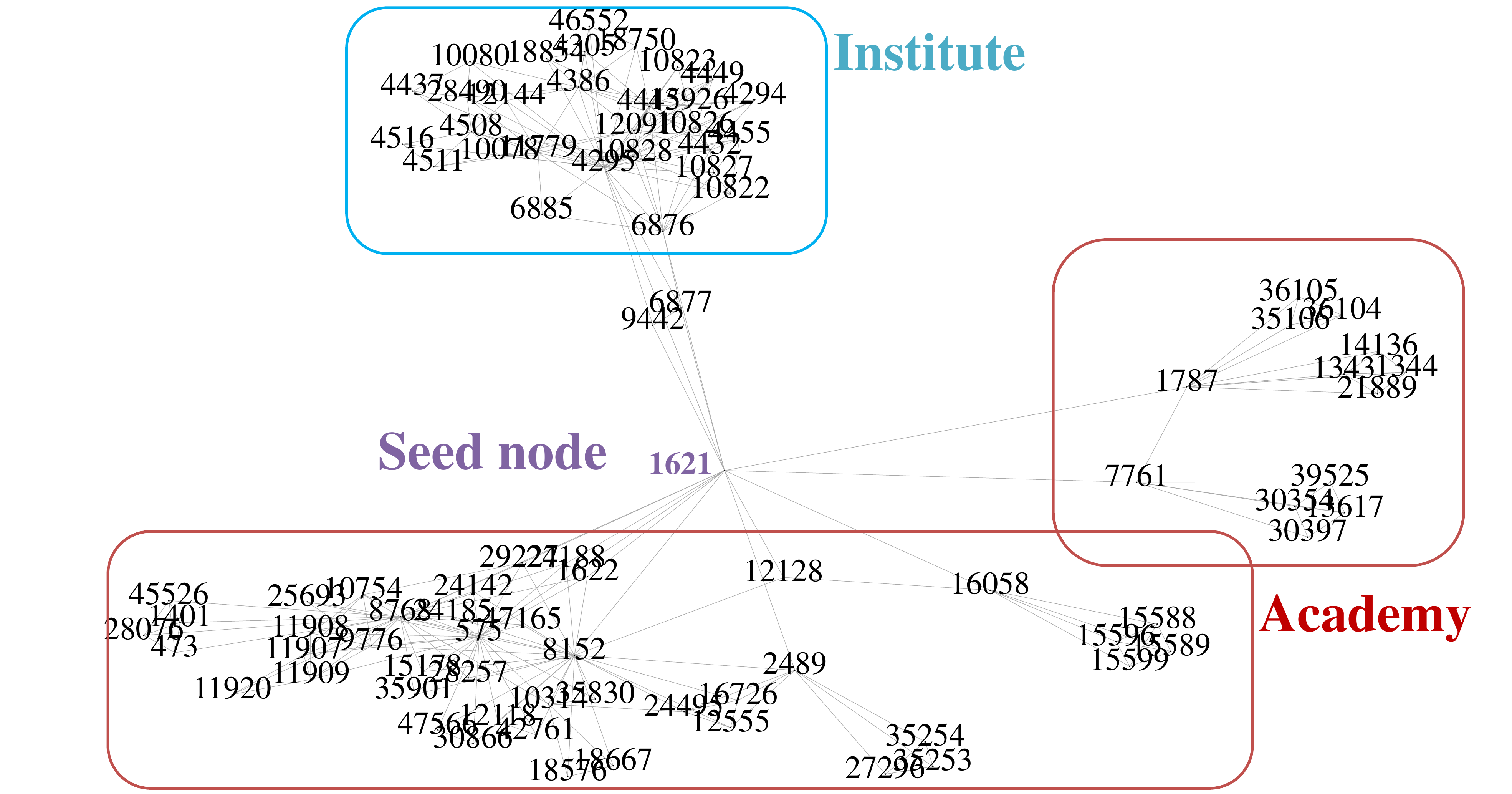}
\vspace{-0.3cm}
\caption{DBLP result when $k=3$}
\vspace{-0.1cm}
\label{fig:case}
\end{figure}
\spara{Case study : DBLP.}
We use DBLP dataset~\cite{kim2014link} for our case study. We generate a co-authorship network where an edge of two authors is generated if they publish at least $3$ papers together. We set $k=3$ and $r=2$. 
Figure~\ref{fig:case} shows the result of the case study. We notice that the author $1621$ is the center of {\SEG} and all the nodes are reachable from $1621$ within $2$ hops. We observe that there are two major sub-communities : (1) two research groups of Academy; (2) a research group of an institute. 

\begin{figure}[h]
\centering
\includegraphics[width=0.9\linewidth]{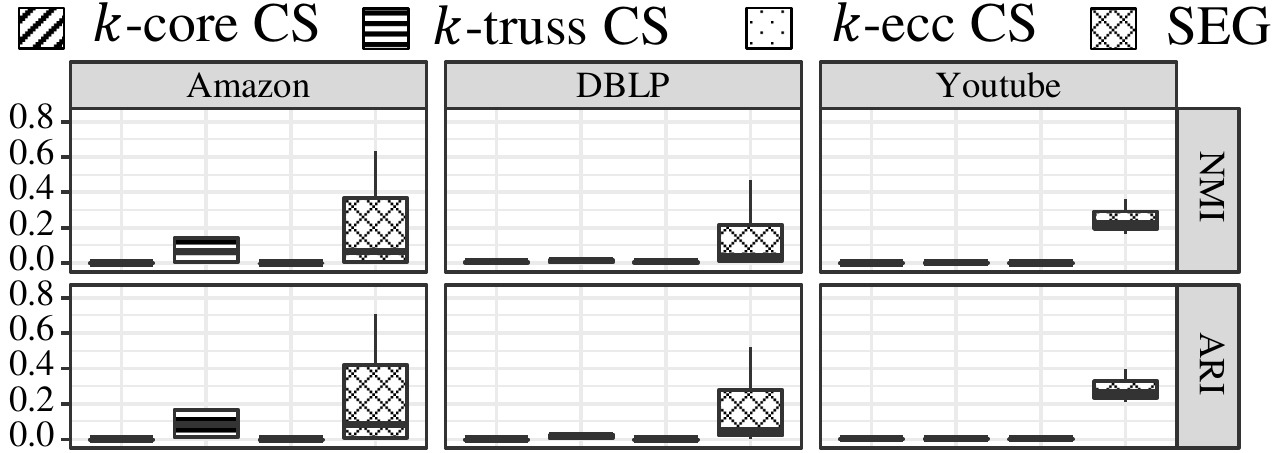}
\caption{Community search experiment}
\vspace{-0.2cm}
\label{fig:cs}
\end{figure}

\spara{Case study : Community Search.}
In this experiment, we demonstrate the results when our {\LUEM} approach is utilized on the community search problem~\cite{sozio2010community,fang2020survey}. 
We use three representative community search models as baseline algorithms : 
\textcolor{black}{
the $k$-core model~\cite{sozio2010community}, $k$-truss model~\cite{huang2014querying}, and $k$-ecc model~\cite{chang2015index} among many community search models~\cite{barbieri2015efficient,huang2014querying,dmcs,wu2015robust,huang2015} since the implementation of the algorithms is publicly available~\cite{fang2020survey}.
}
To verify the accuracy, we use two representative metrics for the community detection problem: NMI(Normalized Mutual Information)~\cite{danon2005comparing} and ARI(Adjusted Rand Index)~\cite{hubert1985comparing}.
Because both measures are designed to identify the best partitions (communities), we consider the community search problem as a binary classification problem to utilize both measures. We use three datasets (Amazon, DBLP, and Youtube~\cite{yang2015defining}) that are reported in Table~\ref{tab:dataset}, which have ground-truth communities. 
Because the ground-truth communities overlap, we compare the identified community with all ground-truth communities containing the query node; then, select the best accuracy value. For each experiment, we randomly picked $10$ query nodes, set $k=3$ and $r=1$, and reported the average and standard deviation. 
Figure~\ref{fig:cs} reports the results of {\LUEM} and two baseline algorithms. Observe that in the Amazon dataset, our model returns a better result than that of the $k$-core model and a result comparable to that of the $k$-truss model. In DBLP and Youtube datasets, observe that our model achieves better accuracy than those of the $k$-core, $k$-truss, and $k$-ecc models. This is because the size of ground-truth communities is relatively small and they normally have a small diameter.

\begin{figure}[h]
\centering
\includegraphics[width=0.95\linewidth]{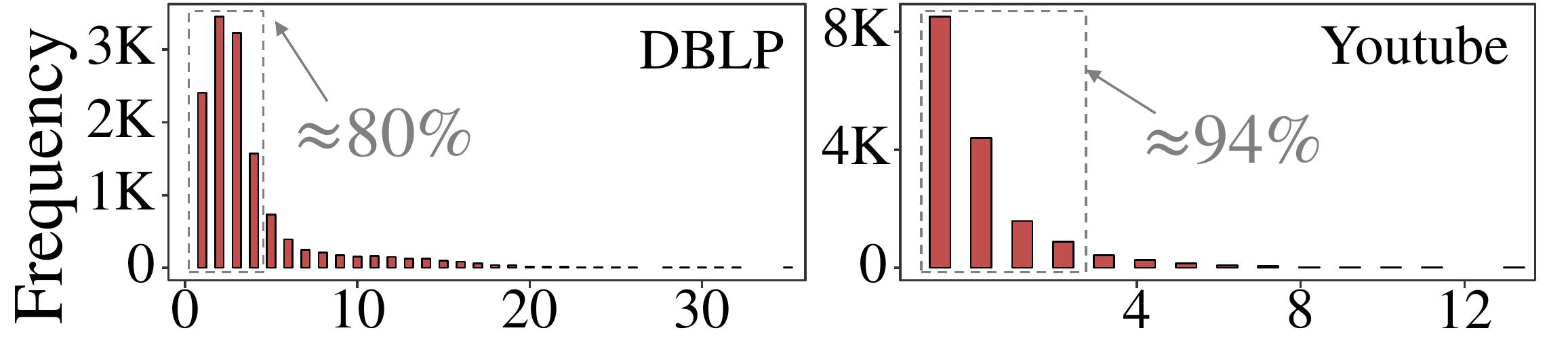}
\vspace{-0.2cm}
\caption{Diameter distribution}
\label{fig:diam}
\end{figure}

Figure~\ref{fig:diam} depicts the   diameter distribution of ground-truth communities in the DBLP and Youtube datasets. Notice that approximately 80\% of communities in DBLP dataset and 95\% of communities in Youtube had a diameter of less than or equal to 5. 
Because our community model is a minimum-degree with \textcolor{black}{a} diameter-bounded community model, we can obtain a better accuracy score if the size of the ground-truth communities are very small.

\section{Related Work}\label{sec:relatedwork}

\subsection{Influence maximization}

The influence maximization (IM) problem~\cite{li2018influence}, which is a key problem in social network analysis, has recently gained much attention owing to its potential value. Given a graph $G=(V,E)$, the influence maximization problem aims to find a set of $k$ seed nodes while maximizing the expected number of users influenced by the seed nodes. A representative application of IM is viral marketing~\cite{chen2010scalable,huang2019community}.
Kempe et al.~\cite{kempe2003maximizing} first modeled the influence maximization  problem for two fundamental information diffusion models: (1) an independent cascade model and (2) a linear threshold model.
Most existing IM problems can be solved by applying the greedy hill-climbing algorithm. Because the objective function of the IM problem is non-negative monotone submodular, if the best node to maximize the number of expected influenced users \textcolor{black}{are} chosen, a $1-\frac{1}{e}$ approximation ratio is achieved.  

\subsection{Anchored and collapsed $k$-core}

The anchored $k$-core problem was first proposed by Bhawalkar et al.~\cite{bhawalkar2012preventing}. Given a graph $G$, integer $k$, and budget $b$, the anchored $k$-core problem involves identifying $b$ anchor nodes to maximize the number of engaged users. Note that $b$ anchor nodes are fixed to be engaged, even if they do not have sufficient neighbor nodes to be engaged. In \cite{bhawalkar2012preventing}, the authors showed that  when $k=2$, polynomial-time algorithms could be designed. When $k\geq 3$, the inapproximability results were proved. Zhang et al.~\cite{zhang2017olak} proposed the OLAK algorithm for the anchored $k$-core problem. They proposed an onion layer structure to significantly reduce the search space for this problem. 
They also proposed early termination and pruning techniques to improve efficiency. Linghu et al.\cite{linghu2020global} proposed an anchored coreness problem. Instead of maximizing the number of engaged users, they focused on the coreness gain by anchoring the $b$ nodes. They presented that the anchored coreness problem was also NP-hard and proposed a pruning search space technique and method to reuse the intermediate results to improve efficiency.
Cai et al. \cite{cai2020anchored} proposed an  attributed community engagement problem that considers the attributes of users and community cohesiveness. 
They aimed to identify $l$ anchored users that can induce a maximal expanded community, namely, Anchored Vertex set Exploration (AVE) problem. They demonstrated that, when $k\geq 3$, the AVE problem was NP-hard. To solve this problem, they proposed the filter-verify algorithm with early termination and pruning techniques. In \cite{laishram2020residual}, the authors proposed residual core maximization, an algorithm for the anchored $k$-core problem. They selected anchored nodes based on two strategies: residual degree and anchor score. Moreover, they showed that the result is close to the optimal solution.

The collapsed $k$-core problem is firstly proposed by Zhang et al.~\cite{zhang2017finding}. They focused on the problem of finding $b$ collapsers to minimize the number of engaged users in a social network when removing $b$ collapsers. The identified $b$ anchor nodes can be considered as important users for maintaining their engagement in a network. Luo et al.~\cite{luo2021parameterized} proved that when $k\geq 3$, the collapsed $k$-core problem was W[P] hard. Zhang et al.~\cite{zhang2018finding} introduced a collapsed $k$-truss problem.

\subsection{$k$-core Decomposition and Its Variation}

The $k$-core is widely used for finding cohesive subgraphs in a graph. The definition of the $k$-core~\cite{seidman1983network} is as follows. Given a graph $G=(V,E)$ and integer $k$, the $k$-core, denoted by $D_k$, is a set of nodes of which every node has at least $k$ neighbor nodes in $D_k$. This $k$-core is unique and has containment relationship, i.e., $D_k \subseteq D_{k-1}$ when $k-1\geq 1$. The core index of a node is $k$ if it belongs to the $k$-core but not to $(k+1)$-core. Note that $k$-core is a set of nodes of which its core index is larger than or equal to $k$. 
Batagelj et al.~\cite{batagelj2003m} propose an exact and efficient algorithm to find the core index. Sariyuce et al.~\cite{sariyuce2016incremental} focuses on  incremental $k$-core problem in a dynamic graph. 
Bonchi et al.~\cite{bonchi2019distance} formulate distance-generalized $k$-core named $(k,h)$-core. 
Zhang et al.~\cite{zhang2020exploring} formulate $(k,p)$-core by considering the degree ratio in $k$-core.
Zhang et al.~\cite{zhang2018discovering} formulate $(k,s)$-core by unifying $k$-truss and $k$-core. 
There are several variations of $k$-core in signed networks~\cite{giatsidis2014quantifying}, directed networks~\cite{giatsidis2013d}, weighted graph~\cite{eidsaa2013s,galimberti2017core}, temporal graph~\cite{wu2015core}, multi-layer graph~\cite{galimberti2017core}, bipartite graph~\cite{ding2017efficient}, and uncertain graph~\cite{bonchi2014core}. 
To get more details, please refer to this nice survey paper~\cite{malliaros2020core}.

\section{Conclusion}

In this study, we formulate a novel problem called \underline{L}ocal \underline{U}ser \underline{E}ngagement \underline{M}aximization (\LUEM) by integrating the local user engagement with graph distance. 
We prove that the {\LUEM} problem is NP-hard and its objective function is monotonic submodular.
To solve this problem, we propose two approximation algorithms and an efficient heuristic algorithm that preserves effectiveness. 
To demonstrate the superiority of our algorithms, we conducted extensive experiments using real-world and synthetic networks. 
As a future research direction, we will consider a dynamic environment to find $b$ seed nodes. For example, the graph can be changed dynamically, or the end user may change the parameters online.

\appendix

\section{Efficiency and effectiveness test when $r=1$}\label{app:r1}

\begin{figure}[ht]
\centering
\includegraphics[width=0.99\linewidth]{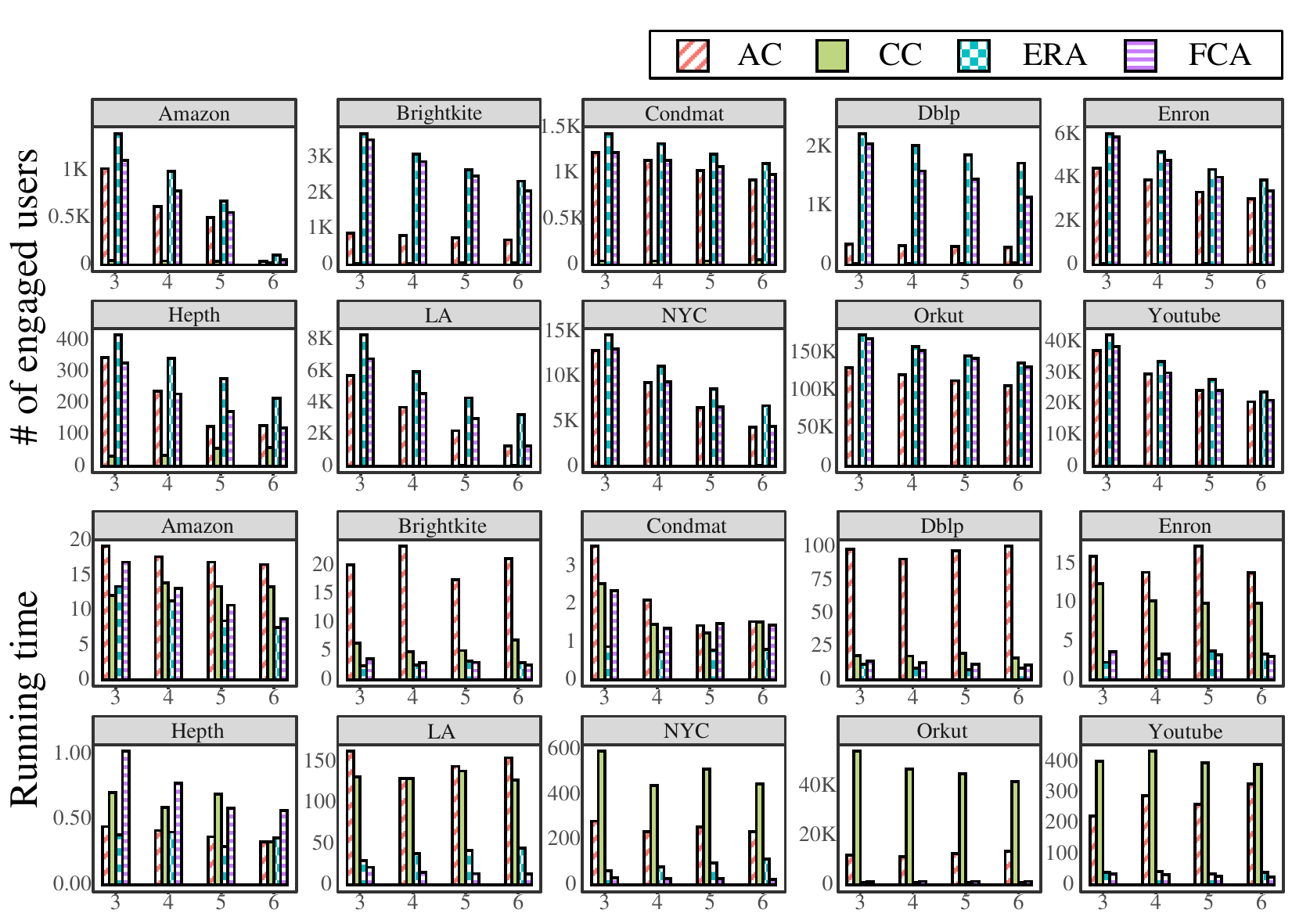}
\vspace{-0.4cm}
\caption{$r=1$}
\vspace{-0.2cm}
\label{fig:r1}
\end{figure}
Figure~\ref{fig:r1} shows effectiveness and efficiency results for $r=1$. Observe that the results have trends similar to those shown in  Figures~\ref{fig:effectiveness_test} and Figures~\ref{fig:efficiency_test}. 
One remarkable difference is that our {\FCA} is slightly slower than the baseline algorithms because it requires some initialization steps for approximation.  When $r\geq 2$, the initialization time is not significantly affected.

\begin{figure}[t]
\centering
\includegraphics[width=0.99\linewidth]{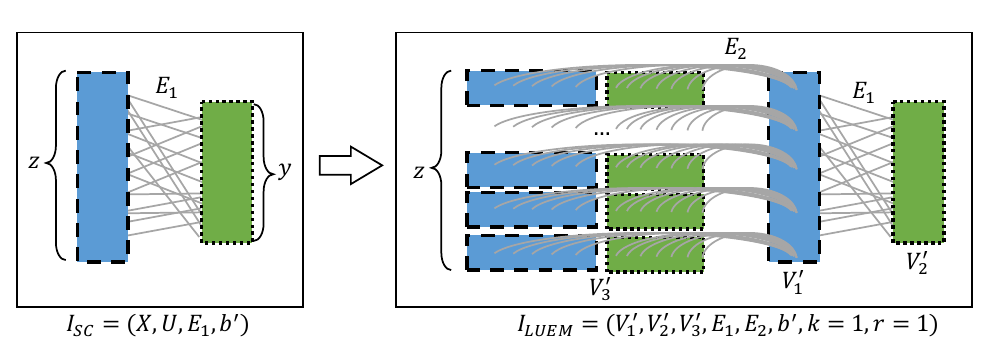}
\vspace{-0.25cm}
\caption{Reduction procedure}
\label{fig:np}
\end{figure}

\section{Proof of Theorem~\ref{theorem:NP}}\label{app:proof}

\textcolor{black}{
To show the NP-hard of a specific problem, it is required to show an example of \textcolor{black}{an} NP-hard problem that can be reduced to the problem that we want to verify the hardness. 
To show the hardness of {\LUEM} problem, we utilize the set-cover problem~\cite{bernhard2008combinatorial}. %
The set-cover problem is defined as follows: given a set of elements $U=\{u_1, u_2, \cdots, u_y\}$, $z$ subsets $X=\{X_1, X_2, \cdots X_z\}$ of $U$, a set of assignments $E_1$ from $X$ to $U$ where $X_i\in X$ contains the element $u_j\in U$, and a parameter $b' > 0$, it returns YES if there exists a set of $b'$ subsets whose union is the same with $U$. Otherwise, it returns NO. 
To present a reduction, we suppose that we have a solution of {\LUEM}. 
Next, we start with an arbitrary instance of the set-cover problem, and show that our {\LUEM} problem can be utilized to solve the set-cover problem.
Let suppose that we have an instance of the set-cover problem $I_{SC}=(X,U,E_1,b')$. By utilizing the instance $I_{SC}$, we can construct a new instance $I_{\LUEM}$. $I_{\LUEM}$ contains three key elements: (1) $V_1'$ is the same with $X$; (2) $V_2'$ is the same with $U$; and (3) $V_3'$ is $z$ sets of $(y+z)$ entities. A subset $X_j$ is connected to the element $i$ if $X_j$ contains the element $i$. We denote this relationship as $E_1$. A subset $X_j$ is connected to the $j$-th $(y+z)$ entities in $V_3'$. We set $b=b'$, $k=1$, and $r=1$. The newly generated instance and graph $G'$ can be checked in Figure~\ref{fig:np}. 
}

\textcolor{black}{
For every entity in $V_2'$, it can have at most $z$ degree since the number of entities in $V_1'$ is $z$. All the entities in $V_3'$ has exactly one neighbor entity in $V_1'$. Interestingly, a node in $V_1'$ has at least $z+y$ degree since it is connected to the nodes in $V_3'$. Since $k=1$ and $r=1$, the size of {\SEG} of node $v$ is the size of the ego-network of node $v$ in $G'$.
}

\section{Proof of Lemma~\ref{lemma:ineq}}\label{app:lemma_proof}

We denote $OPT \setminus \mathcal{S} = \{v_1, v_2, \ldots, v_l\}$  
where $l \leq b$. We then easily notice that

\textcolor{black}{
\begin{align}
    \rho&(OPT) \leq \rho(OPT \cup \mathcal{S}) \nonumber \\
             &= \rho(\mathcal{S}) + \sum_{j=1}^{l} [\rho(\mathcal{S} \cup \{v_1, \ldots, v_j\} ) - \rho(\mathcal{S} \cup \{v_1,\ldots, v_{j-1}\} )   ] \nonumber   \\ 
             & \leq \rho(\mathcal{S}) + \sum_{j=1}^{l} [\rho(\mathcal{S}\cup \{v_j\})  - \rho(\mathcal{S})]   \\
             & \leq \rho(\mathcal{S}) + \sum_{j=1}^{l} \max_{x\in V} [\rho(\mathcal{S}\cup \{x\})  - \rho(\mathcal{S})]\nonumber   \\
             & \leq \rho(\mathcal{S}) + b\max_{x\in V} [\rho(\mathcal{S}\cup \{x\})  - \rho(\mathcal{S})] \nonumber 
\end{align} 
}

It indicates that the following inequality holds. 
\begin{align}
\begin{aligned}
    & \rho(OPT)\leq \rho(\mathcal{S}) + b\max_{x\in V} [\rho(\mathcal{S}\cup \{x\})  - \rho(\mathcal{S})] \\
    & \frac{1}{b}(\rho(OPT)-\rho(\mathcal{S})) \leq \max_{x\in V} [\rho(\mathcal{S}\cup \{x\})  - \rho(\mathcal{S})] 
\end{aligned}
\end{align} 
Hence, our Lemma~\ref{lemma:ineq} holds. \qedhere

\section{Proof of Theorem~\ref{theorem:ratio}}\label{app:ratio_proof}

Let denote $\mathcal{K}^i$ as the solution of our algorithm at the end of the iteration $i\leq b$. Then, we reuse the Lemma~\ref{lemma:ineq}. 
\begin{align}
\rho(\mathcal{K}^i) &- \rho(\mathcal{K}^{i-1})  \geq \frac{1}{b} (\rho(OPT) - \rho(\mathcal{K}^{i-1})) \nonumber \\
\rho(\mathcal{K}^i) &- \rho(\mathcal{K}^{i-1})  \geq  \frac{1}{b}\rho(OPT) - \frac{1}{b}\rho(\mathcal{K}^{i-1}) \nonumber \\
\rho(\mathcal{K}^i) &\geq \frac{1}{b} \rho(OPT) + (1- \frac{1}{b}) \rho(\mathcal{K}^{i-1})  \\
&\geq \frac{1}{b} \rho(OPT) + (1- \frac{1}{b}) [\frac{1}{b}\rho(OPT) + (1-\frac{1}{b})\rho(\mathcal{K}^{i-2})] \nonumber \\
&\geq \dfrac{1- (1 - \frac{1}{b})^b}{b(1 - (1 - \frac{1}{b}))} \rho(OPT) = [1- (1 - \frac{1}{b})^b] \rho(OPT) \\
& \geq (1 - \frac{1}{e}) \rho(OPT)\nonumber
\end{align}

Therefore, we notice that $\rho(\mathcal{K}^i) \geq (1-\frac{1}{e}) \rho(OPT)$. 

\textcolor{black}{
In $G'$, we can find a solution of {\LUEM} with $b=b'$, $k=1$, and $r=1$.
To find a solution, we must choose the nodes in $V_1'$ since it can make at least $1+y+z$ nodes to be engaged as we discussed. Thus, if we find a solution of {\LUEM} in $G'$ and it covers all the nodes in $V_2'$, we can solve the set-cover problem with parameter $b'$. Hence, {\LUEM} problem is NP-hard. 
}

\bibliographystyle{ACM-Reference-Format}
\bibliography{sample-base}

\end{document}